\documentclass[11pt, a4paper]{article}
\RequirePackage{fullpage}

\usepackage{amsthm,amsmath,amssymb}
\usepackage[usenames,dvipsnames]{color}
\usepackage[pdftex,breaklinks,colorlinks,
    citecolor={BlueViolet}, linkcolor={Blue},urlcolor=Maroon]{hyperref}
\usepackage[final,expansion=alltext,protrusion=true]{microtype}
\usepackage{booktabs}
\newcommand{\comment}[1]{\textbackslash\!\!\textbackslash {\em #1}}
\usepackage{graphicx,enumerate}
\usepackage{array}

\usepackage{tikz}
\usetikzlibrary{graphs,graphs.standard,quotes}
\usetikzlibrary{shapes.geometric}
\usetikzlibrary{arrows,shapes,positioning}
\usetikzlibrary{calc}

\tikzstyle{vertex} = [{circle,blue,draw,fill=black!50,inner sep=1pt}]
\tikzstyle{pathv} = [{regular polygon,regular polygon sides=3,draw,fill=violet!50,inner sep=1.2pt}]
\tikzstyle{basev} = [{violet,draw,fill=violet!50,inner sep=2pt}]
\tikzstyle{arc} = [-latex, dashed, double distance=2pt]

\usepackage{subcaption}

\newtheorem{redrule}{Reduction Rule}
\newtheorem{exrule}{Exchange Rule}
\newtheorem{theorem}{Theorem}
\newtheorem{lemma}{Lemma}[section]
\newtheorem{corollary}[lemma]{Corollary}
\newtheorem{proposition}[lemma]{Proposition}
\theoremstyle{definition}
\newtheorem{definition}[lemma]{Definition}

\title{A $5k$-vertex Kernel for $P_2$-packing}

\author{Wenjun Li\thanks{School of Computer and Communication Engineering, Changsha University of Science and Technology, Changsha, China. \texttt{liwenjun@csu.edu.cn}.  Supported by the National Natural Science Foundation of China (NSFC) under grant 61872048 and the Research Foundation of Education Bureau of Hunan Province under grant 21B0305.}
\and
Junjie Ye\thanks{Department of Computing, Hong Kong Polytechnic University, Hong Kong, China. \texttt{\{junjie.ye, yixin.cao\}@polyu.edu.hk}.  {Supported in part by the Hong Kong Research Grants Council (RGC) under grants 15221420 and 15201317, and the National Natural Science Foundation of China (NSFC) under grant 61972330.}}
\and
Yixin Cao\footnotemark[2]
}
\date{}

\begin{document}
\maketitle
\begin{abstract}
  The $P_2$-packing problem asks whether a graph contains $k$ vertex-disjoint (not necessarily induced) paths each of length two.  We continue the study of its kernelization algorithms, and develop a $5k$-vertex kernel.
\end{abstract}

\section{Introduction}
Packing problems make one of the most important family of problems in combinatorial optimization.  One example is $H$-packing for a fixed graph $H$, i.e., to find the maximum number of vertex-disjoint copies of $H$ from a graph $G$.
It is trivial when $H$ consists of a single vertex, and it is the well-known maximum matching problem, which can be solved in polynomial time, when $H$ is an edge.  The problem can be easily reduced to the maximum matching problem when each component of $H$ has at most two vertices.
The smallest $H$ on which the $H$-packing problem is NP-complete is $P_2$, the graph on three vertices and two edges \cite{kirkpatrick-83-graph-factor-problems}.   The $P_2$-packing problem is thus a natural starting point of investigating $H$-packing problems in general, and has been extensively studied \cite{kaneko-01, kelmans-04, kelmans-11-p2-packing, monnot-07-path-partition}. 

In the parameterized setting, the $P_2$-packing problem asks whether a graph $G$ contains $k$ vertex-disjoint $P_2$'s.  Recall that given an instance $(G, k)$, a {\em kernelization algorithm} produces in polynomial time an equivalent instance $(G', k')$---$(G, k)$ is a yes-instance if and only if $(G', k')$ is a yes-instance---such that $k' \leq k$.  The size of $G'$ is upper bounded by some function of $k'$, and $(G', k')$ is a \textit{polynomial kernel} when the function is a polynomial function.
Prieto and Sloper~\cite{prieto-06-packing-stars} first developed a $15 k$-vertex kernel for the $P_2$-packing problem, which were improved to $7 k$ \cite{wang-10-p2-packing} and then $6 k$ \cite{chen-19-kernel-packing-covering}.  We further improve it to $5 k$.

\begin{theorem}\label{thm:main}
	The $P_2$-packing problem has a $5k$-vertex kernel.
\end{theorem}

Our improvement, although modest, is a solid step toward the ultimate goal of this line of research, a kernel of only $3 k$ vertices.  Note that the problem remains NP-hard when $G$ has exactly $3k$ vertices.  Indeed, what Kirkpatrick and Hell~\cite{kirkpatrick-83-graph-factor-problems} proved is the NP-hardness of deciding whether a graph can be partitioned into vertex-disjoint $P_2$'s.  The existence of a $3k$-vertex kernel for the $P_2$-packing problem would indicate that it is morally equivalent to the $P_2$-partition problem.  Moreover, our algorithm implies directly an approximation algorithm of ratio $5/3$; a $4$-vertex kernel, provided that it satisfies certain properties, would imply an approximation algorithm of ratio $4/3$, better than the best known ratio $4/3 + \epsilon$ \cite{cygan-13-3d-matching}.  We remark that the problem is MAX-SNP-hard~\cite{kann-94}, and remains so even on bipartite graphs with maximum degree three \cite{monnot-07-path-partition}.  So if there would exist a $3k$-vertex kernel, it had to use a different technique.

We also note that there are efforts on a simplified version of the problem: Chang et al.~\cite{chang-14-p2packing} claimed a $5k$-vertex kernel for the problem on net-free graphs, which however contains a critical bug, according to Xiao and Kou~\cite{xiao-17-cvd}.

A very natural tool for the problem is a generalization of the well-studied crown structure~\cite{fellows-03-crown-reduction}.  Prieto and Sloper~\cite{prieto-06-packing-stars} used it in the first kernel, and all later work follows suit.  If we have a set $C$ of vertices such that (1) $G[C]$ does not contain any $P_2$, and (2) there are precisely $|N(C)|$ vertex-disjoint $P_2$'s in the subgraph induced by $N(C)\cup C$, then we may take these $|N(C)|$ paths and consider the subgraph $G - (N(C)\cup C)$.  This remains our main reduction rule; the difficulty, hence one of our contributions, is how to find such a structure if one exists.

As all the previous kernelization algorithms for the problem, we start from finding a maximal $P_2$-packing $\cal P$ in a greedy way.  Let $V({\cal P})$ denote the vertices on paths in $\cal P$, and we call other vertices, i.e., $V(G)\setminus V({\cal P})$, extra vertices.  Note that each component in $G - V({\cal P})$ contains at most two vertices, and we may remove it from the graph if it is not adjacent to $V({\cal P})$.  We may assume henceforth that each component in $G - V({\cal P})$ is connected with $V({\cal P})$.  By some classic results from matching theory, a reducible structure can be identified as long as the number of extra vertices, hence the number of components in $G - V({\cal P})$, is large enough.  This leads to the first kernel of $15 k$ vertices \cite{prieto-06-packing-stars}, and is the starting point of all later results.

A similar scheme is employed in later work.  It tries to find a $P_2$-packing larger than $\cal P$ using local search; once the local search gets stuck, a careful study of the configuration may reveal new reducible structures.
Observing that a pair of adjacent extra vertices is more helpful for this approach than two nonadjacent ones, Wang et al.~\cite{wang-10-p2-packing} used two simple exchange rules to consolidate extra vertices.  For example, if the two ends of a path on five vertices are extra vertices (i.e., the three vertices in the middle are picked to be a path in $\cal P$), they would change it so that the two extra vertices are adjacent.
The key idea of Chen et al.~\cite{chen-19-kernel-packing-covering} is that extra vertices adjacent to the ends of a path in $\cal P$ are usually more helpful than those adjacent to the middle vertex of the path.

Our local search procedure is more systematic and comprehensive; it actually subsumes observations from both Wang et al.~\cite{wang-10-p2-packing} and Chen et al.~\cite{chen-19-kernel-packing-covering}.
After the standard opening step, if we are not able to find reducible structure, we assign the extra vertices to paths in $\cal P$ such that each path receives a small number of them.  Each path, together with the assigned vertices, defines a \textit{unit}.
Based on their local structures, we put units with at least five vertices into two categories. 
We introduce several nontrivial exchange rules to migrate vertices from ``large" units to ``small" units.  Their applications may lead to (1) a larger $P_2$-packing than $\cal P$, or (2) a reducible structure, whereupon we repeat the procedure with a larger number of units or a smaller graph, respectively.  After all of our rules have been exhaustively applied, a unit contains at most six vertices, and the number of six-vertex units is upper bounded by the number of small units with four or three vertices.  The bound on the size of the kernel follows immediately.

\section{Preliminaries}
All graphs discussed in this paper are undirected and simple.
The vertex set and edge set of a graph $G$ are denoted by $V(G)$ and $E(G)$ respectively.
For a set $U\subseteq V(G)$ of vertices, we denote by $G[U]$ the subgraph induced by $U$, whose vertex set is $U$ and whose edge set comprises all edges of $G$ with both ends in $U$.   We use $G - U$ as a shorthand for $G[V(G) \setminus U]$, and it is further simplified as $G - v$ when $U$ consists of only one vertex $v$.  A {\em component} is a maximal connected induced subgraph, and an {\em edge component} is a component on two vertices.

\begin{redrule}\label{rul:component}
  If a component $C$ of $G$ has at most 6 vertices, delete $C$ and decrease $k$ by the maximum number of vertex-disjoint $P_2$'s in $C$.
\end{redrule}

The following technical definition would be crucial for our main reduction rule.
\begin{definition}
  Let $C$ be a set of vertices and $N(C) = \{v_1, v_2, \dots, v_{\ell}\}$.  We say that $C$ is a \emph{reducible set} of $G$ if the maximum degree in $G[C]$ is at most one and one of the following holds.
  \begin{enumerate}[(i)]
  \item There are $\ell$ distinct edge components $\{C_1, \dots, C_{\ell}\}$ in $G[C]$ such that $v_i$ is adjacent to $C_i$ for $1 \le i \le \ell$.
  \item There are $2\ell$ distinct components $\{C_1, \dots, C_{2\ell}\}$ in $G[C]$ such that $v_i$ is adjacent to $C_{2i-1}$ and $C_{2i}$ for $1 \le i \le \ell$.
  \end{enumerate}
\end{definition}
For readers familiar with previous work, a remark is worthwhile here.  Our reducible set is a generalization of the well-known crown decomposition \cite{fellows-03-crown-reduction}.  Our definition (i) coincides with the ``fat crown'' defined in \cite{prieto-06-packing-stars}.  Our definition (ii) coincides with the ``double crown'' defined in \cite{prieto-06-packing-stars} when each component of $G[C]$ is a single vertex.  In definition (ii), however, we allow a mixture of single-vertex components and edge components.  As a matter of fact, one may define the reducible set in a way that an edge component is regarded as two single-vertex components.  This definition would work for our algorithm and might reveal more reducible sets.  However, it would slightly complicate our presentation without helping our analysis in the worst case, and hence we choose to use the simpler one.
\begin{redrule}\label{rul:p2-crown}
  If there is a reducible set $C$, delete $N(C) \cup C$ and decrease $k$ by $|N(C)|$.
\end{redrule}

The safeness of Rule \ref{rul:p2-crown} can be easily adapted from a similar proof of Prieto and Sloper~\cite{prieto-06-packing-stars}.  For the sake of completeness, we include it here.  We use $\mathtt{opt}(G)$ to denote the maximum number of vertex-disjoint $P_2$'s in graph $G$.
\begin{lemma}\label{lem:p2-crown}
  If $C$ is a reducible set of graph $G$, then $\mathtt{opt}(G) = \mathtt{opt}(G- (N(C) \cup C)) + |N(C)|$.
\end{lemma}
\begin{proof}
  Let $A = N(C)$ and $G' = G - (A \cup C)$.  For each vertex $v\in A$, we can pick in $G[C]$ an edge component or two vertices from two components to form a $P_2$. Thus we have $|A|$ vertex-disjoint $P_2$'s using only vertices in $A \cup C$.  Together with a maximum $P_2$-packing of $G'$, we have $\mathtt{opt}(G) \ge \mathtt{opt}(G') + |A|$.

  On the other hand, any $P_2$-packing of $G$ contains at most $|A|$ vertex-disjoint $P_2$'s involving vertices in $A$. Hence $\mathtt{opt}(G) \le \mathtt{opt}(G - A) + |A|$. By definition, the maximum degree in $G[C]$ is at most one, and hence vertices of $C$ participate in no $P_2$ in $G - A$.  Therefore, $\mathtt{opt}(G - A) = \mathtt{opt}(G')$ and $\mathtt{opt}(G) \le \mathtt{opt}(G') + |A|$.
\end{proof}

To identify reducible sets, we will rely on tools from matching theory.  In several steps of our algorithm, we will construct an auxiliary bipartite graph $B$; to avoid confusion, we use nodes to refer to elements in $V(B)$.  
The two sides of $B$ are denoted by $L$ and $R$.
We will use the Hopcroft-Karp algorithm: 
\begin{lemma}[\cite{hopcroft-73-bipartite-matching}]\label{lem:hopcroft-matching}
  Given a bipartite graph $B$, we can find in polynomial time a matching saturating $L$ or a set $L' \subseteq L$ such that there is a matching between $N(L')$ and $L'$ that saturates $N(L')$.
\end{lemma}

\section{The unit partition}

We first find a maximal $P_2$-packing $\cal P$ of the input graph $G$, and use these paths as ``bases'' to partition $V(G)$ into $|\cal P|$ units.  We may return a trivial true instance and terminate the algorithm when $|{\cal P}|\ge k$.  Henceforth we assume $|{\cal P}|< k$.  We then locally change the units so that they satisfy certain properties.  During the process, if we find (1) a $P_2$-packing larger than $\cal P$, or (2) a reducible set, then we restart the procedure with a new $P_2$-packing, or a new graph respectively.

Denote by $V({\cal P})$ the set of vertices in the paths in $\cal P$.  The maximality of $\cal P$ guarantees that each component of the subgraph $G - V({\cal P})$ is either a single vertex or an edge.  We construct an auxiliary bipartite graph $B_1$ as follows:
\begin{itemize}
\item for each component $C$ of $G - V({\cal P})$, introduce a node $u_C$ into $L$;
\item for each vertex $v \in V({\cal P})$, introduce two nodes $v^1, v^2$ into $R$; and
\item add edges $u_Cv^1$ and $u_Cv^2$ if vertex $v$ is adjacent to $V(C)$ in $G$.
\end{itemize}

\begin{lemma}\label{lem:matching}
  If there is no matching of $B_1$ saturating all nodes in $L$, then we can find in polynomial time a reducible set.
\end{lemma}
\begin{proof}
  By Lemma~\ref{lem:hopcroft-matching}, we find in polynomial time a subset $L'\subseteq L$ such that there is a matching of $B_1$ between $N_{B_1}(L')$ and $L'$ that saturates all nodes in $N_{B_1}(L')$. Let $C'$ be the vertices in the components represented by nodes in $L'$, and let $A'$ be the set of vertices represented by nodes in $N_{B_1}(L')$.  We claim that $C'$ is a reducible set. Note that for each vertex $v \in V({\cal P})$, the set $N_{B_1}(L')$ contains either both or neither of $\{v^1, v^2\}$. For each $v \in A'$, the two components in $G[C']$, whose nodes are matched to $v^1$ and $v^2$, are adjacent to $v$. By the construction of $B_1$, $G[C']$ has maximum degree at most one and $N(C') = A'$. Hence $C'$ is a reducible set.
\end{proof}

In the following, we may assume that we have a matching $M$ of $B_1$ saturating all nodes in $L$.  For a path $P\in {\cal P}$ on vertices $u, v, w$, we create a \emph{unit} that contains $u, v, w$, and all vertices in those components matched to nodes $u^1, u^2, v^1, v^2, w^1, w^2$ by $M$.  Abusing the notation, we also use unit to refer to the subgraph induced by it, which is always connected.  The path $P$ is the {\em base path} of this unit.  Since all nodes in $L$ are matched in $M$, the collection of units is a partition of the vertex set $V(G)$.
If each unit has five or fewer vertices, then $|V(G)|< 5 k$ and we are done.
By construction, a unit may contain at most six components (each of at most two vertices) of $G - V({\cal P})$, hence up to $3 + 6 * 2 = 15$ vertices.  There are a prohibitive number of graphs on 15 vertices; the following exchange rule excludes most of them from our further consideration.

\begin{exrule}\label{rule:single-unit}
  If a unit contains two vertex-disjoint $P_2$'s, then we find a larger maximal $P_2$-packing than $\cal P$ as follows.  Let $P^0$ be the base path of the unit, and $P^1$ and $P^2$ the two vertex-disjoint $P_2$'s, then we take ${\cal P}\setminus \{P^0\} \cup\{P^1, P^2\}$.
\end{exrule}

It is easy to see that the partition produced by Exchange Rule~\ref{rule:single-unit} is valid: Vertices in different parts are disjoint. Once Exchange Rule~\ref{rule:single-unit} is applied, we restart the procedure with a new $P_2$-packing and hence a new unit partition.  In the following we may assume that a unit has precisely one vertex-disjoint $P_2$.

Let us motivate our main technical definitions with a simple example.  The two graphs in Figure~\ref{fig:two-units} comprise the same pair of units, (one cycle and one star, both on five vertices,) but are connected by different edges in between.  The two units behave very differently: While the first unit (cycle) is willing to sacrifice any pair of adjacent vertices, there is a vertex not affordable to lose by the second unit (star).

\begin{figure}[h]
  \centering
  \begin{subfigure}[b]{.35\linewidth}
    \centering
    \begin{tikzpicture}[scale=.5]
      \begin{scope}[every node/.style={vertex},]
        \node (a) at (1,1.5) {};
        \node (b) at (3,1.5) {};
      \end{scope}
      \begin{scope}[every node/.style={pathv}]
        \node (u) at (0,0) {};
        \node (v) at (2,0) {};
        \node (w) at (4,0) {};
      \end{scope}
      \draw[thick] (u) -- (v) -- (w);
      \draw (u) -- (a) -- (b) -- (w);
      \begin{scope}[every node/.style={vertex}, shift={(6, 0)}]
        \node (c) at (1,1.5) {};
        \node (d) at (3,1.5) {};
      \end{scope}
      \begin{scope}[every node/.style={pathv}, shift={(6, 0)}]
        \node (x) at (0,0) {};
        \node (y) at (2,0) {};
        \node (z) at (4,0) {};
      \end{scope}
      \draw[thick] (x) -- (y) -- (z);
      \draw (c) -- (y) -- (d);

      \draw[blue, thick, dashed] (w) -- (c);
    \end{tikzpicture}
    \caption{}
  \end{subfigure}  
  \qquad\qquad
  \begin{subfigure}[b]{.35\linewidth}
    \centering
    \begin{tikzpicture}[scale=.5]
      \begin{scope}[every node/.style={vertex},]
        \node (a) at (1,1.5) {};
        \node (b) at (3,1.5) {};
      \end{scope}
      \begin{scope}[every node/.style={pathv}]
        \node (u) at (0,0) {};
        \node (v) at (2,0) {};
        \node (w) at (4,0) {};
      \end{scope}
      \draw[thick] (u) -- (v) -- (w);
      \draw (u) -- (a) -- (b) -- (w);
      \begin{scope}[every node/.style={vertex}, shift={(6, 0)}]
        \node (c) at (1,1.5) {};
        \node (d) at (3,1.5) {};
      \end{scope}
      \begin{scope}[every node/.style={pathv}, shift={(6, 0)}]
        \node (x) at (0,0) {};
        \node (y) at (2,0) {};
        \node (z) at (4,0) {};
      \end{scope}
      \draw[thick] (x) -- (y) -- (z);
      \draw (c) -- (y) -- (d);

      \draw[blue, thick, dashed] (b) -- (y);
    \end{tikzpicture}
    \caption{}
  \end{subfigure}  
  \caption{Each graph consists of two units on five vertices.  In a unit, the three triangle vertices and the thick edges make the base path.   Other solid edges are inside the unit, while a dashed edge crosses two units.  We can find three vertex-disjoint $P_2$'s from (a), but not (b).}
  \label{fig:two-units}
\end{figure}

We say that a unit is \emph{democratic} if it contains one of the graphs in Figure~\ref{fig:unit1} as a spanning subgraph.  We call a democratic unit a net-, pan-, $C_5$-, or bull-unit if it contains net, pan, $C_5$, or bull but none of the previous ones as a subgraph. This order ensures, among others, that a bull-unit has to be an induced bull: We can find a pan or a $C_5$ if any additional edge is added to a bull.  (Indeed, only pan-unit can have extra edges, but we only need this property for bull-units.)  A bull-unit contains a unique vertex of degree 2, which we call the \emph{nose} of the bull-unit.
It is easy to check that \textit{there remains a $P_2$ in a democratic unit after any vertex removed}.

 \begin{figure}[h]
   \centering
   \begin{subfigure}[b]{.15\linewidth}
    \centering
     \begin{tikzpicture}[every node/.style={vertex},scale=.5]
       \begin{scope}
         \node (a) at (0,3.5) {};
         \node (b) at (0,2) {};
         \node (f) at (2.2,-0.6) {};
       \end{scope}
       \begin{scope}
         \node (c) at (-1,0.5) {};
         \node (d) at (1,0.5) {};
         \node (e) at (-2.2,-0.6) {};
       \end{scope}
       \draw (a) -- (b) -- (c) (d) -- (f);
       \draw (b) -- (d);
       \draw (d) -- (c) -- (e);
     \end{tikzpicture}
    \caption{net}{}
  \end{subfigure}  
   \qquad
   \begin{subfigure}[b]{.15\linewidth}
    \centering
     \begin{tikzpicture}[every node/.style={vertex},scale=.5]
       \node (a) at (0,3.2) {};
       \node (b) at (0, 1.6) {};
       \node (c) at (-1.25, 0.3) {};
       \node (e) at (1.25, 0.3) {};
       \node (d) at (0, -1) {};
       \draw (a) -- (b) -- (c) -- (d) -- (e) -- (b);
     \end{tikzpicture}
    \caption{pan}
  \end{subfigure}  
   \qquad
   \begin{subfigure}[b]{.15\linewidth}
    \centering
     \begin{tikzpicture}[every node/.style={vertex},scale=.5]
       \node (a) at (0.65,0) {};
       \node (b) at (0,1.65) {};
       \node (c) at (1.5,3) {};
       \node (d) at (3,1.65) {};
       \node (e) at (2.35,0) {};
       \draw (a)-- (b)--(c)--(d)--(e)--(a);
     \end{tikzpicture}
    \caption{$C_5$}
  \end{subfigure}  
   \qquad
   \begin{subfigure}[b]{.15\linewidth}
    \centering
     \begin{tikzpicture}[every node/.style={vertex},scale=.5]
       \node (c) at (-0.9,0.5) {};
       \node (b) at (0.9,0.5) {};
       \node["$v$" below] (a) at (0,-1) {};
       \node (d) at (-2.2981,1.92) {};
       \node (e) at (2.2981,1.92) {};
       \draw (d) -- (c) -- (a) -- (b) -- (e);
       \draw (b) -- (c);
     \end{tikzpicture}
    \caption{bull}
  \end{subfigure}  
   \caption{Four subgraphs characterizing democratic units.  A pan-unit may contain extra edges not shown here, while the other three cannot.  The vertex $v$ in (d) is the nose of the bull-unit.}
   \label{fig:unit1}
 \end{figure}

A unit that has more than four vertices but is not democratic is called \emph{despotic}.  For example, the second unit in Figure~\ref{fig:two-units}(a) is despotic.  In that example, even though the only degree-four vertex in it cannot be removed from the unit, all the other four vertices are dispensable.  For such a unit, it does not make much sense to distinguish the two vertices in the base path and the other two vertices: We may take any $P_2$ from this unit as the base path, but it clearly has to contain the degree-4 vertex.

 \begin{figure}[h]
  \centering
   \begin{subfigure}[b]{.12\linewidth}
    \centering
    \begin{tikzpicture}[scale=.5]
      \node[basev] (root) at (1.5,0) {};
      \begin{scope}[every node/.style={vertex}]
      \node (a) at (0,1.5) {};
      \node (a1) at (0,3) {};
      \node (b) at (1,1.5) {};
      \node (b1) at (1,3) {};
      \node (c) at (2,1.5) {};
      \node (c1) at (2,3) {};
      \node (d) at (3,1.5) {};
      \node (d1) at (3,3) {};
      \end{scope}
      \draw (root) -- (a) -- (a1);
      \draw (root) -- (b) -- (b1);
      \draw (root) -- (c) -- (c1);		
      \draw (root) -- (d) -- (d1);
    \end{tikzpicture}
    \caption{$(4, 0)$}
  \end{subfigure}  
   \quad
   \begin{subfigure}[b]{.12\linewidth}
    \centering
    \begin{tikzpicture}[scale=.5]
      \node[basev] (root) at (1.5,0) {};
      \begin{scope}[every node/.style={vertex}]
      \node (a) at (0,1.5) {};
      \node (a1) at (0,3) {};
      \node (b) at (1,1.5) {};
      \node (b1) at (1,3) {};
      \node (c) at (2,1.5) {};
      \node (c1) at (2,3) {};
      \node (d) at (3,1.5) {};
      \end{scope}
      \draw (root) -- (a) -- (a1);
      \draw (root) -- (b) -- (b1);
      \draw (root) -- (c) -- (c1);		
      \draw (root) -- (d);
    \end{tikzpicture}
    \caption{$(3, 1)$}
  \end{subfigure}  
   \quad
   \begin{subfigure}[b]{.1\linewidth}
    \centering
    \begin{tikzpicture}[scale=.5]
      \node[basev] (root) at (1, 0) {};
      \begin{scope}[every node/.style={vertex}]
      \node (a) at (0,1.5) {};
      \node (a1) at (0,3) {};
      \node (b) at (1,1.5) {};
      \node (b1) at (1,3) {};
      \node (c) at (2,1.5) {};
      \node (c1) at (2,3) {};
      \end{scope}
      \draw (root) -- (a) -- (a1);
      \draw (root) -- (b) -- (b1);
      \draw (root) -- (c) -- (c1);		
    \end{tikzpicture}
    \caption{$(3, 0)$}
  \end{subfigure}  
   \quad
   \begin{subfigure}[b]{.12\linewidth}
    \centering
    \begin{tikzpicture}[scale=.5]
      \node[basev] (root) at (1.5,0) {};
      \begin{scope}[every node/.style={vertex}]
      \node (a) at (0,1.5) {};
      \node (a1) at (0,3) {};
      \node (b) at (1,1.5) {};
      \node (b1) at (1,3) {};
      \node (c) at (2,1.5) {};
      \node (d) at (3,1.5) {};
      \end{scope}
      \draw (root) -- (a) -- (a1);
      \draw (root) -- (b) -- (b1);
      \draw (root) -- (c);		
      \draw (root) -- (d);
    \end{tikzpicture}
    \caption{$(2, 2)$}
  \end{subfigure}  
   \quad
   \begin{subfigure}[b]{.1\linewidth}
    \centering
    \begin{tikzpicture}[scale=.5]
      \node[basev] (root) at (1., 0) {};
      \begin{scope}[every node/.style={vertex}]
      \node (a) at (0,1.5) {};
      \node (a1) at (0,3) {};
      \node (b) at (1,1.5) {};
      \node (b1) at (1,3) {};
      \node (c) at (2,1.5) {};
      \end{scope}
      \draw (root) -- (a) -- (a1);
      \draw (root) -- (b) -- (b1);
      \draw (root) -- (c);		
    \end{tikzpicture}
    \caption{$(2, 1)$}
  \end{subfigure}  
   \quad
   \begin{subfigure}[b]{.1\linewidth}
    \centering
    \begin{tikzpicture}[scale=.5]
      \node[basev] (root) at (.75,0) {};
      \begin{scope}[every node/.style={vertex}]
      \node (a) at (0,1.5) {};
      \node (a1) at (0,3) {};
      \node (b) at (1.5, 1.5) {};
      \node (b1) at (1.5, 3) {};
      \end{scope}
      \draw (root) -- (a) -- (a1);
      \draw (root) -- (b) -- (b1);
    \end{tikzpicture}
    \caption{$(2, 0)$}
  \end{subfigure}  

   \begin{subfigure}[b]{.12\linewidth}
    \centering
    \begin{tikzpicture}[scale=.5]
      \node[basev] (root) at (2, 0) {};
      \begin{scope}[every node/.style={vertex}]
      \node (a) at (0,1.5) {};
      \node (a1) at (0,3) {};
      \node (b) at (1,1.5) {};
      \node (c) at (2,1.5) {};
      \node (d) at (3,1.5) {};
      \node (e) at (4,1.5) {};
    \end{scope}
    \draw (root) -- (a) -- (a1);
      \draw (root) -- (b);
      \draw (root) -- (c);
      \draw (root) -- (d);
      \draw (root) -- (e);
    \end{tikzpicture}
    \caption{$(1, 4)$}
  \end{subfigure}  
   \quad
   \begin{subfigure}[b]{.12\linewidth}
    \centering
    \begin{tikzpicture}[scale=.5]
      \node[basev] (root) at (1.5, 0) {};
      \begin{scope}[every node/.style={vertex}]
      \node (a) at (0,1.5) {};
      \node (a1) at (0,3) {};
      \node (b) at (1,1.5) {};
      \node (c) at (2,1.5) {};
      \node (d) at (3,1.5) {};
    \end{scope}
      \draw (root) -- (a) -- (a1);
      \draw (root) -- (b);
      \draw (root) -- (c);
      \draw (root) -- (d);
    \end{tikzpicture}
    \caption{$(1, 3)$}
  \end{subfigure}  
   \quad
   \begin{subfigure}[b]{.1\linewidth}
    \centering
    \begin{tikzpicture}[scale=.5]
      \node[basev] (root) at (1., 0) {};
      \begin{scope}[every node/.style={vertex}]
      \node (a) at (0,1.5) {};
      \node (a1) at (0,3) {};
      \node (b) at (1,1.5) {};
      \node (c) at (2,1.5) {};
    \end{scope}
      \draw (root) -- (a) -- (a1);
      \draw (root) -- (b);
      \draw (root) -- (c);
    \end{tikzpicture}
    \caption{$(1, 2)$}
  \end{subfigure}  
   \quad
   \begin{subfigure}[b]{.12\linewidth}
    \centering
    \begin{tikzpicture}[scale=.5]
      \node[basev] (root) at (1.5, 0) {};
      \begin{scope}[every node/.style={vertex}]
      \node (a) at (0,1.5) {};
      \node (b) at (1,1.5) {};
      \node (c) at (2,1.5) {};
      \node (d) at (3,1.5) {};
    \end{scope}
      \draw (root) -- (a);
      \draw (root) -- (b);
      \draw (root) -- (c);
      \draw (root) -- (d);
    \end{tikzpicture}
    \caption{$(0, 4)$}
  \end{subfigure}  
  \caption{Ten subgraphs characterizing despotic units.  The square vertices are the core vertices, while the round vertices make the peripheral.  Each edge component of the peripheral is a twig, and each single-vertex component is a leaf.  Each unit in the first row has at least two twigs, while the second row has at most one.  For the units with the same number of twigs, they are ordered by the number of leaves (hence the total number of vertices).  Note that a despotic unit may contain extra edges not shown here.}
  \label{fig:unit2}
\end{figure}

Each graph $F$ in Figure \ref{fig:unit2} has a special vertex $v$ (the square vertex at the bottom) such that its removal leaves a graph of maximum degree at most one.   In other words, each component of $F - v$ is an edge or an isolated vertex; we call them a \emph{twig} and a \emph{leaf} respectively.  We label $F$ by a pair $(a, b)$, which are the numbers of, respectively, twigs and leaves of $F$.
{We say that a despotic unit is an $(a, b)$-unit if it can be made, by deleting edges, graph $(a, b)$ but not any graph $(a', b')$ with $a' > a$.}
A consequence of enforcing this order (of maximizing twigs) is that there cannot be any edge between two leaves in any unit: For example, if an edge is added to connect the two leaves of graph $(2, 2)$, then it also contains graph $(3, 0)$ as a subgraph.   For a unit $U$, we also use $d_1(U)$ and $d_2(U)$ to denote the numbers of, respectively, twigs and leaves; i.e., $d_1(U) = a$ and $d_2(U) = b$ when $U$ is an ($a, b$)-unit.  The special vertex is the \emph{core}, while all other vertices (including twigs and leaves) the \emph{peripheral}, of the unit.  It would be clear that there exists still a $P_2$ in a despotic unit after we delete a twig or leaf from it.

In passing we should mention that although we draw only one edge between the core and a twig, both vertices in the twig can be adjacent to the core.  For our algorithm, we always consider a twig as a whole.

We are left with the \emph{small units} (of three or four vertices), which turn out to be singular in our algorithm.  Although they are the smallest units, great care is needed to deal with them.   Recall that our aim is to bound the number of vertices by summing all units in the final unit partition of the final graph; hence we would like to maximize the number of small units.  In this sense, the role of a small unit as an ``exporter'' would be marginal, if possible, and hence we do not categorize them into many types.  We abuse the notation to denote them in a  similar way as despotic units. 
A four-vertex unit is a $(0, 3)$-unit if it has precisely three edges and all of them share a common end, and a $(0, 1)$-unit otherwise.  A three-vertex unit is a $(0, 0)$-unit, disregarding whether it has two or three edges.  See Figure~\ref{fig:small-units} for an illustration of small units. Note that $(0,1)$-unit and $(0,0)$-unit are special in the sense that they have three core vertices.\footnote{It seems to make more sense to use type $(0, 2)$ instead of $(0, 1)$ for Figure~\ref{fig:small-units}(b), because both ends of the path can be ``sacrificed.''  The reason is that we want to avoid edges between leafs from the same unit; see Proposition~\ref{pro:extra-edges}.  Of course, we may introduce a new type of $C_4$-units, but that would complicate our analysis with no direct benefit.}

\begin{figure}[h]
  \centering
  \begin{subfigure}[b]{.15\linewidth}
    \centering
    \begin{tikzpicture}[scale=.5]
      \node[basev] (root) at (0, 0) {};
      \begin{scope}[every node/.style={vertex}]
      \node (a) at (0, 1.5) {};
      \node (b) at (-.75, 1.5) {};
      \node (c) at (.75, 1.5) {};
    \end{scope}
      \node at (-1., 1.5) {};
      \node at (1, 1.5) {};
      \draw (root) -- (a);
      \draw (b) -- (root) -- (c);
    \end{tikzpicture}
    \caption{$(0, 3)$}
  \end{subfigure}  
  \qquad
  \begin{subfigure}[b]{.15\linewidth}
    \centering
    \begin{tikzpicture}[scale=.5]
      \begin{scope}[every node/.style={basev}]
      \node (root) at (0, 0) {};
      \node (b) at (1., 0) {};
      \node (c) at (2, 0) {};
    \end{scope}
      \node[vertex] (a) at (0, 1.5) {};
    \draw (root) -- (a);
      \draw (root) -- (b) -- (c);
    \end{tikzpicture}
    \caption{$(0, 1)$}
  \end{subfigure}  
  \qquad
  \begin{subfigure}[b]{.15\linewidth}
    \centering
    \begin{tikzpicture}[every node/.style={basev}, scale=.5]
      \node (root) at (1,0) {};
      \node (b) at (0,0) {};
      \node (c) at (2,0) {};
      \draw (b) -- (root) -- (c);
    \end{tikzpicture}
    \caption{$(0, 0)$}
  \end{subfigure}  
  \caption{Three subgraphs characterizing small units. The square vertices are the core vertices, while the round vertices make the peripheral.  Note that a ($0, 3$)-unit cannot contain extra edges, while the other two can.}
  \label{fig:small-units}
\end{figure}

The following is straightforward from Exchange Rule~\ref{rule:single-unit} and the definitions of despotic and small units.
\begin{proposition}\label{pro:extra-edges}
  Let $U$ be an $(a,b)$-unit.  If Exchange Rule~\ref{rule:single-unit} is not applicable to $U$, then there is no edge between different twigs or leaves of $U$.
\end{proposition}
\begin{proof}
  This is very obvious for small units: It is by definition for $(0,3)$-units; and there are at most one leaf and no twig in $(0,1)$- and $(0,0)$-units.  Suppose for contradiction that there exists an edge $e$ between different twigs or leaves of an $(a,b)$-unit $U$.  By the definition of $(a,b)$-units, at least one end of $e$ comes from a twig, hence $a \ge 1$.   If $U$ is a $(1, 3)$-unit or contains at least seven vertices, then $U$ contains two vertex-disjoint $P_2$'s, contradicting the assumption.  Thus, $U$ has to be a $(2, 1)$-, $(2,0)$- or $(1,2)$-unit.  If $U$ is a $(2, 1)$-unit with an extra edge but does not contain two vertex-disjoint $P_2$'s, then it has to be a net, hence a net-unit; see Figure~\ref{fig:extra-edges}(a).  Likewise, if $U$ is a $(2, 0)$- or $(1,2)$-unit with an extra edge, then it is a democratic unit; see Figure~\ref{fig:extra-edges}(b,c).  Therefore, we always end with a contradiction, and this concludes the proof.
\end{proof}

\begin{figure}[h]
  \centering
  \begin{subfigure}[b]{.4\linewidth}
    \centering
    \begin{tikzpicture}[scale=.5]
      \node[basev] (root) at (1., 0) {};
      \begin{scope}[every node/.style={vertex}]
      \node (a) at (0,1.5) {};
      \node (a1) at (0,3) {};
      \node (b) at (1,1.5) {};
      \node (b1) at (1,3) {};
      \node (c) at (2,1.5) {};
      \end{scope}
      \draw (root) -- (a) -- (a1);
      \draw (root) -- (b) -- (b1);
      \draw (root) -- (c);
      \draw (a) -- (b);
    \end{tikzpicture}
    \begin{tikzpicture}[scale=.5]
      \node[basev] (root) at (1., 0) {};
      \begin{scope}[every node/.style={vertex}]
      \node (a) at (0,1.5) {};
      \node (a1) at (0,3) {};
      \node (b) at (1,1.5) {};
      \node (b1) at (1,3) {};
      \node (c) at (2,1.5) {};
      \end{scope}
      \draw (root) -- (a) -- (a1);
      \draw (root) -- (b) -- (b1);
      \draw (root) -- (c);
      \draw[thick] (c) -- (root) -- (a) (a1) -- (b1) -- (b);
    \end{tikzpicture}
    \begin{tikzpicture}[scale=.5]
      \node[basev] (root) at (1., 0) {};
      \begin{scope}[every node/.style={vertex}]
      \node (a) at (0,1.5) {};
      \node (a1) at (0,3) {};
      \node (b) at (1,1.5) {};
      \node (b1) at (1,3) {};
      \node (c) at (2,1.5) {};
      \end{scope}
      \draw (root) -- (a) -- (a1);
      \draw (root) -- (b) -- (b1);
      \draw (root) -- (c);		
      \draw[thick] (c) -- (root) -- (a) (a1) -- (b) -- (b1);
    \end{tikzpicture}
    \begin{tikzpicture}[scale=.5]
      \node[basev] (root) at (1., 0) {};
      \begin{scope}[every node/.style={vertex}]
      \node (a) at (0,1.5) {};
      \node (a1) at (0,3) {};
      \node (b) at (1,1.5) {};
      \node (b1) at (1,3) {};
      \node (c) at (2,1.5) {};
      \end{scope}
      \draw (root) -- (a) -- (a1);
      \draw (root) -- (b) -- (b1);
      \draw (root) -- (c);		
      \draw[thick] (root) -- (a) -- (a1) (c) -- (b) -- (b1);
    \end{tikzpicture}
    \begin{tikzpicture}[scale=.5]
      \node[basev] (root) at (1., 0) {};
      \begin{scope}[every node/.style={vertex}]
      \node (a) at (0,1.5) {};
      \node (a1) at (0,3) {};
      \node (b) at (1,1.5) {};
      \node (b1) at (1,3) {};
      \node (c) at (2,1.5) {};
      \end{scope}
      \draw (root) -- (a) -- (a1);
      \draw (root) -- (b) -- (b1);
      \draw (root) -- (c);		
      \draw[thick] (root) -- (a) -- (a1) (c) -- (b1) -- (b);
    \end{tikzpicture}
    \caption{$(2, 1)$}
  \end{subfigure}  
  \qquad
  \begin{subfigure}[b]{.2\linewidth}
    \centering
    \begin{tikzpicture}[scale=.5]
      \node[basev] (root) at (.75,0) {};
      \begin{scope}[every node/.style={vertex}]
      \node (a) at (0,1.5) {};
      \node (a1) at (0,3) {};
      \node (b) at (1.5, 1.5) {};
      \node (b1) at (1.5, 3) {};
      \end{scope}
      \draw (root) -- (a) -- (a1) -- (b1);
      \draw (root) -- (b) -- (b1);
    \end{tikzpicture}
    \begin{tikzpicture}[scale=.5]
      \node[basev] (root) at (.75,0) {};
      \begin{scope}[every node/.style={vertex}]
      \node (a) at (0,1.5) {};
      \node (a1) at (0,3) {};
      \node (b) at (1.5, 1.5) {};
      \node (b1) at (1.5, 3) {};
      \end{scope}
      \draw (root) -- (a) -- (a1) (a) -- (b);
      \draw (root) -- (b) -- (b1);
    \end{tikzpicture}
    \begin{tikzpicture}[scale=.5]
      \node[basev] (root) at (.75,0) {};
      \begin{scope}[every node/.style={vertex}]
      \node (a) at (0,1.5) {};
      \node (a1) at (0,3) {};
      \node (b) at (1.5, 1.5) {};
      \node (b1) at (1.5, 3) {};
      \end{scope}
      \draw (root) -- (a) -- (a1) -- (b);
      \draw (root) -- (b) -- (b1);
    \end{tikzpicture}
    \caption{$(2, 0)$}
  \end{subfigure}  
  \qquad
  \begin{subfigure}[b]{.15\linewidth}
    \centering
    \begin{tikzpicture}[scale=.5]
      \node[basev] (root) at (1., 0) {};
      \begin{scope}[every node/.style={vertex}]
      \node (a) at (0,1.5) {};
      \node (a1) at (0,3) {};
      \node (b) at (1,1.5) {};
      \node (c) at (2,1.5) {};
    \end{scope}
      \draw (root) -- (a) -- (a1);
      \draw (root) -- (b) -- (a1);
      \draw (root) -- (c);
    \end{tikzpicture}
    \begin{tikzpicture}[scale=.5]
      \node[basev] (root) at (1., 0) {};
      \begin{scope}[every node/.style={vertex}]
      \node (a) at (0,1.5) {};
      \node (a1) at (0,3) {};
      \node (b) at (1,1.5) {};
      \node (c) at (2,1.5) {};
    \end{scope}
      \draw (root) -- (a) -- (a1);
      \draw (root) -- (b) -- (a);
      \draw (root) -- (c);
    \end{tikzpicture}
    \caption{$(1, 2)$}
  \end{subfigure}  
  \caption{Demonstration that a $(2, 1)$-, $(2,0)$- or $(1,2)$-unit cannot have an extra edge between different twigs or leaves.  We either end with a democratic unit, or two disjoint $P_2$'s, shown as thick edges.  Note that there cannot be an edge between the two leaves of a $(1, 2)$-unit because it should otherwise have been classified as a $(2, 0)$-unit.}
  \label{fig:extra-edges}
\end{figure}

An exhaustive case analysis will convince us that we have covered all possible units on which Exchange Rule~\ref{rule:single-unit} is not applicable.   We should point out that a $(1,4)$-unit cannot be part of a unit partition produced as above; they can only be introduced by exchange rules to be presented in the next section.
\begin{proposition}\label{pro:raw-unit-partition}
  Let $U$ be a unit in a unit partition produced from a maximal $P_2$-packing $\cal P$.  If Exchange Rule~\ref{rule:single-unit} is not applicable to $U$, then it must be a democratic, despotic, or a small unit.  Moreover, it cannot be a $(1, 4)$-unit.
\end{proposition}
\begin{proof}
  Let $U$ be a unit and $P$ its base path.  We can find two vertex-disjoint $P_2$'s in the unit if (1) an edge component is matched to one end of $P$ and one component is matched to another vertex of $P$; or (2) two components are matched to one end of $P$ and one component is matched to another vertex of $P$. Disregarding these cases, if Exchange Rule~\ref{rule:single-unit} is not applicable, each unit falls into one of the cases in Table~\ref{tbl:matching}. We can check that (1) each unit on five or more vertices is a democratic/despotic unit, but not a $(1,4)$-unit; (2) other units are small units.
\end{proof}
\begin{table}[ht]
  \centering
  \caption{The configurations of units that contain exactly one $P_2$.  The three triangle vertices make the base path $P_2$, while others (round vertices) make components matched to them.  We only present the edges between a component and the vertex on the base path to which it is matched.
  We enumerate first units without edge components, then units with edge components; both of which are ordered by the number of vertices.} \label{tbl:matching}

\newcolumntype{C}{ >{\centering\arraybackslash} m{1.5cm} }
\newcolumntype{D}{ >{\centering\arraybackslash} m{5.5cm} }
  \scriptsize
  \begin{tabular}{>{\centering\arraybackslash}m{2.2cm} >{\centering\arraybackslash}m{2.2cm} >{\centering\arraybackslash}m{2.2cm} >{\centering\arraybackslash}m{2.2cm} >{\centering\arraybackslash}m{2.2cm} >{\centering\arraybackslash}m{2.2cm}}
    \toprule
    \begin{tikzpicture}
	\begin{scope}[scale=.5]
	\begin{scope}[every node/.style={pathv},]
	\node (a) at (0,0) {};
	\node (b) at (1.5,0) {};
	\node (c) at (3,0) {};
	\end{scope}
	\node at (1.5,1.5) {};
	\end{scope}
	\draw (a) -- (b) -- (c);
    \end{tikzpicture}
    &
    \begin{tikzpicture}
    \begin{scope}[every node/.style={pathv},scale=.5]
    \node (a) at (0,0) {};
    \node (b) at (1.5,0) {};
    \node (c) at (3,0) {};
    \end{scope}
    \begin{scope}[every node/.style={vertex},scale=.5]
    \node (a1) at (0,1.5) {};
    \end{scope}

    \draw (a1) -- (a) -- (b) -- (c);
    \end{tikzpicture}
    &
	\begin{tikzpicture}
	\begin{scope}[every node/.style={pathv},scale=.5]
	\node (a) at (0,0) {};
	\node (b) at (1.5,0) {};
	\node (c) at (3,0) {};
	\end{scope}
	\begin{scope}[every node/.style={vertex},scale=.5]
	\node (b1) at (1.5,1.5) {};
	\end{scope}
	
	\draw (a) -- (b) -- (c);
	\draw (b1) -- (b);
	\end{tikzpicture}
    &
	\begin{tikzpicture}
	\begin{scope}[every node/.style={pathv},scale=.5]
	\node (a) at (0,0) {};
	\node (b) at (1.5,0) {};
	\node (c) at (3,0) {};
	\end{scope}
	\begin{scope}[every node/.style={vertex},scale=.5]
	\node (a1) at (0,1.5) {};
	\node (b1) at (1.5,1.5) {};
	\end{scope}
	
	\draw (a1) -- (a) -- (b) -- (c);
	\draw (b1) -- (b);
	\end{tikzpicture}
    &
	\begin{tikzpicture}
	\begin{scope}[every node/.style={pathv},scale=.5]
	\node (a) at (0,0) {};
	\node (b) at (1.5,0) {};
	\node (c) at (3,0) {};
	\end{scope}
	\begin{scope}[every node/.style={vertex},scale=.5]
	\node (a1) at (0,1.5) {};
	\node (c1) at (3,1.5) {};
	\end{scope}
	
	\draw (a1) -- (a) -- (b) -- (c) -- (c1);
	\end{tikzpicture}
    &
    \begin{tikzpicture}
	\begin{scope}[every node/.style={pathv},scale=.5]
	\node (a) at (0,0) {};
	\node (b) at (1.5,0) {};
	\node (c) at (3,0) {};
	\end{scope}
	\begin{scope}[every node/.style={vertex},scale=.5]
	\node (a1) at (-0.5,1.5) {};
	\node (a2) at (0.5,1.5) {};
	\end{scope}
	\draw (a) -- (b) -- (c);
	\draw (a1) -- (a) -- (a2);
	\end{tikzpicture}
    \\
    $(0,0)$
    &
    $(0,1)$
    &
    $(0,3)$, $(0,1)$
    &
    pan, bull, $(2,0)$, $(1,2)$
    &
    pan, bull, $(2,0)$
    &
    {pan, bull, $(1,2)$}
    \\
	\midrule

    \begin{tikzpicture}
    \begin{scope}[every node/.style={pathv},scale=.5]
    \node (a) at (0,0) {};
    \node (b) at (1.5,0) {};
    \node (c) at (3,0) {};
    \end{scope}
    \begin{scope}[every node/.style={vertex},scale=.5]
    \node (b1) at (1,1.5) {};
    \node (b2) at (2,1.5) {};
    \end{scope}

    \draw (a) -- (b) -- (c);
    \draw (b1) -- (b) -- (b2);
    \end{tikzpicture}
    &
    \begin{tikzpicture}
    \begin{scope}[every node/.style={pathv},scale=.5]
    \node (a) at (0,0) {};
    \node (b) at (1.5,0) {};
    \node (c) at (3,0) {};
    \end{scope}
    \begin{scope}[every node/.style={vertex},scale=.5]
    \node (a1) at (0,1.5) {};
    \node (b1) at (1.5,1.5) {};
    \node (c1) at (3,1.5) {};
    \end{scope}

    \draw (a1) -- (a) -- (b) -- (c) -- (c1);
    \draw (b1) -- (b);
    \end{tikzpicture}
    &
    \begin{tikzpicture}
    \begin{scope}[every node/.style={pathv},scale=.5]
    \node (a) at (0,0) {};
    \node (b) at (1.5,0) {};
    \node (c) at (3,0) {};
    \end{scope}
    \begin{scope}[every node/.style={vertex},scale=.5]
    \node (a1) at (0,1.5) {};
    \node (b1) at (1,1.5) {};
    \node (b2) at (2,1.5) {};
    \end{scope}

    \draw (a1) -- (a) -- (b) -- (c);
    \draw (b1) -- (b) -- (b2);
    \end{tikzpicture}
    &
    \begin{tikzpicture}
    \begin{scope}[every node/.style={pathv},scale=.5]
    \node (a) at (0,0) {};
    \node (b) at (1.5,0) {};
    \node (c) at (3,0) {};
    \end{scope}
    \begin{scope}[every node/.style={vertex},scale=.5]
    \node (a1) at (0,1.5) {};
    \node (b1) at (1,1.5) {};
    \node (b2) at (2,1.5) {};
    \node (c1) at (3,1.5) {};
    \end{scope}

    \draw (a1) -- (a) -- (b) -- (c) -- (c1);
    \draw (b1) -- (b) -- (b2);
    \end{tikzpicture}
    &
    &
    \\
    pan, $(2,0)$, $(1,2)$, $(0,4)$
    &
    net, $(2,1)$
    &
    $(2,1)$, $(1,3)$
    &
    $(2,2)$
    &
    &
    \\
    \midrule

	\begin{tikzpicture}
	\begin{scope}[every node/.style={pathv},scale=.5]
	\node (a) at (0,0) {};
	\node (b) at (1.5,0) {};
	\node (c) at (3,0) {};
	\end{scope}
	\node () at (-0.2, 0) {};
	\begin{scope}[every node/.style={vertex},scale=.5]
	\node (a1) at (0,1.5) {};
	\node (a11) at (0,3) {};
	\end{scope}
	\draw (a) -- (b) -- (c);
	\draw (a11) -- (a1) -- (a);
	\end{tikzpicture}
	&
	\begin{tikzpicture}
	\begin{scope}[every node/.style={pathv},scale=.5]
	\node (a) at (0,0) {};
	\node (b) at (1.5,0) {};
	\node (c) at (3,0) {};
	\end{scope}
	\begin{scope}[every node/.style={vertex},scale=.5]
	\node (b1) at (1.5,1.5) {};
	\node (b2) at (1.5,3) {};
	\end{scope}
	
	\draw (a) -- (b) -- (c);
	\draw (b) -- (b1) -- (b2);
	\end{tikzpicture}
	&
    \begin{tikzpicture}
	\begin{scope}[every node/.style={pathv},scale=.5]
	\node (a) at (0,0) {};
	\node (b) at (1.5,0) {};
	\node (c) at (3,0) {};
	\end{scope}
	\begin{scope}[every node/.style={vertex},scale=.5]
	\node (a1) at (-0.5,1.5) {};
	\node (a11) at (-0.5,3) {};
	\node (a2) at (0.5,1.5) {};
	\end{scope}
	\draw (a) -- (b) -- (c);
	\draw (a11) -- (a1) -- (a) -- (a2);
	\end{tikzpicture}
    &
	\begin{tikzpicture}
	\begin{scope}[every node/.style={pathv},scale=.5]
	\node (a) at (0,0) {};
	\node (b) at (1.5,0) {};
	\node (c) at (3,0) {};
	\end{scope}
	\begin{scope}[every node/.style={vertex},scale=.5]
	\node (b1) at (1,1.5) {};
	\node (b11) at (1,3) {};
	\node (b2) at (2,1.5) {};
	\end{scope}
	
	\draw (a) -- (b) -- (c);
	\draw (b11) -- (b1) -- (b) -- (b2);
	\end{tikzpicture}
    &
	\begin{tikzpicture}
	\begin{scope}[every node/.style={pathv},scale=.5]
	\node (a) at (0,0) {};
	\node (b) at (1.5,0) {};
	\node (c) at (3,0) {};
	\end{scope}
	\begin{scope}[every node/.style={vertex},scale=.5]
	\node (a1) at (0,1.5) {};
	\node (b1) at (1.5,1.5) {};
	\node (b2) at (1.5,3) {};
	\end{scope}
	
	\draw (a1) -- (a) -- (b) -- (c);
	\draw (b) -- (b1) -- (b2);
	\end{tikzpicture}
    &
    \begin{tikzpicture}
	\begin{scope}[every node/.style={pathv},scale=.5]
	\node (a) at (0,0) {};
	\node (b) at (1.5,0) {};
	\node (c) at (3,0) {};
	\end{scope}
	\begin{scope}[every node/.style={vertex},scale=.5]
	\node (a1) at (0,1.5) {};
	\node (b1) at (1,1.5) {};
	\node (b11) at (1,3) {};
	\node (b2) at (2,1.5) {};
	\end{scope}
	
	\draw (a1) -- (a) -- (b) -- (c);
	\draw (b11) -- (b1) -- (b) -- (b2);
	\end{tikzpicture}
    \\
    pan, $C_5$, bull, $(2,0)$
    &
    pan, bull, $(2,0)$, $(1,2)$
    &
    net, $(2,1)$
    &
    $(2,1)$, $(1,3)$
    &
    net, $(2,1)$
    &
    $(3,0)$, $(2,2)$
    \\
    \midrule

    \begin{tikzpicture}
	\begin{scope}[every node/.style={pathv},scale=.5]
	\node (a) at (0,0) {};
	\node (b) at (1.5,0) {};
	\node (c) at (3,0) {};
	\end{scope}
	\begin{scope}[every node/.style={vertex},scale=.5]
	\node (a1) at (0,1.5) {};
	\node (b1) at (1.5,1.5) {};
	\node (b2) at (1.5,3) {};
	\node (c1) at (3,1.5) {};
	\end{scope}
	
	\draw (a1) -- (a) -- (b) -- (c) -- (c1);
	\draw (b) -- (b1) -- (b2);
	\end{tikzpicture}
    &
    \begin{tikzpicture}
	\begin{scope}[every node/.style={pathv},scale=.5]
	\node (a) at (0,0) {};
	\node (b) at (1.5,0) {};
	\node (c) at (3,0) {};
	\end{scope}
	\begin{scope}[every node/.style={vertex},scale=.5]
	\node (a1) at (-0.5,1.5) {};
	\node (a11) at (-0.5,3) {};
	\node (a2) at (0.5,1.5) {};
	\node (a21) at (0.5,3) {};
	\end{scope}
	
	\draw (a1) -- (a) -- (b) -- (c);
	\draw (a11) -- (a1) -- (a) -- (a2) -- (a21);
	\end{tikzpicture}
    &
	\begin{tikzpicture}
	\begin{scope}[every node/.style={pathv},scale=.5]
	\node (a) at (0,0) {};
	\node (b) at (1.5,0) {};
	\node (c) at (3,0) {};
	\end{scope}
	\begin{scope}[every node/.style={vertex},scale=.5]
	\node (b1) at (1,1.5) {};
	\node (b11) at (1,3) {};
	\node (b2) at (2,1.5) {};
	\node (b21) at (2,3) {};
	\end{scope}
	
	\draw (a) -- (b) -- (c);
	\draw (b11) -- (b1) -- (b) -- (b2) -- (b21);
	\end{tikzpicture}
    &
	\begin{tikzpicture}
	\begin{scope}[every node/.style={pathv},scale=.5]
	\node (a) at (0,0) {};
	\node (b) at (1.5,0) {};
	\node (c) at (3,0) {};
	\end{scope}
	\begin{scope}[every node/.style={vertex},scale=.5]
	\node (a1) at (0,1.5) {};
	\node (b1) at (1,1.5) {};
	\node (b11) at (1,3) {};
	\node (b2) at (2,1.5) {};
	\node (c1) at (3,1.5) {};
	\end{scope}
	
	\draw (a1) -- (a) -- (b) -- (c) -- (c1);
	\draw (b11) -- (b1) -- (b) -- (b2);
	\end{tikzpicture}
    &
    \begin{tikzpicture}
	\begin{scope}[every node/.style={pathv},scale=.5]
	\node (a) at (0,0) {};
	\node (b) at (1.5,0) {};
	\node (c) at (3,0) {};
	\end{scope}
	\begin{scope}[every node/.style={vertex},scale=.5]
	\node (a1) at (0,1.5) {};
	\node (b1) at (1,1.5) {};
	\node (b11) at (1,3) {};
	\node (b2) at (2,1.5) {};
	\node (b21) at (2,3) {};
	\end{scope}
	
	\draw (a1) -- (a) -- (b) -- (c);
	\draw (b11) -- (b1) -- (b) -- (b2) -- (b21);
	\end{tikzpicture}
    &
	\begin{tikzpicture}
	\begin{scope}[every node/.style={pathv},scale=.5]
	\node (a) at (0,0) {};
	\node (b) at (1.5,0) {};
	\node (c) at (3,0) {};
	\end{scope}
	\begin{scope}[every node/.style={vertex},scale=.5]
	\node (a1) at (0,1.5) {};
	\node (b1) at (1,1.5) {};
	\node (b11) at (1,3) {};
	\node (b2) at (2,1.5) {};
	\node (b21) at (2,3) {};
	\node (c1) at (3,1.5) {};
	\end{scope}
	
	\draw (a1) -- (a) -- (b) -- (c) -- (c1);
	\draw (b11) -- (b1) -- (b) -- (b2) -- (b21);
	\end{tikzpicture}
    \\
    $(3,0)$
    &
    $(3,0)$
    &
    $(3,0)$, $(2,2)$
    &
    $(3,1)$
    &
    $(3,1)$
    &
    $(4,0)$
    \\
    \bottomrule
  \end{tabular}
\end{table}

Recall that we always assign vertices of a component of $G - V({\cal P})$ to the same unit.  As a result, if an edge is crossing two units, then at least one end of this edge is on a base path.  However, we do not have a similar property on core vertices and peripheral vertices: There can be an edge between two peripheral vertices in different units, and an edge between a peripheral vertex and a democratic unit.  We now apply the following operations to partially restore it, whose correctness is straightforward: Note that (1) a despotic/small unit can sacrifice a twig or leaf; and (2) for every vertex of a democratic unit (except the nose of a bull-unit), we can pick an adjacent vertex such that the unit still contains a $P_2$ after removing the two vertices.
Since at most three units are involved, this rule can be applied in polynomial time.
\begin{exrule}\label{rule:two-units}
  If any of the following holds true, we produce a larger $P_2$-packing than $\cal P$.
  \begin{enumerate}[(i)]
  \item\label{item:rule-democratic} There is an edge between a vertex $u_1$ of a democratic unit and a vertex $u_2$ of another unit where $u_1$ is not a nose and $u_2$ is not a core vertex.
  \item There is an edge between the nose of a bull-unit and a twig of a despotic unit.
  \item There is an edge between a twig of a despotic unit and a peripheral vertex of another unit.
  \end{enumerate}
\end{exrule}

In Exchange Rule~\ref{rule:two-units}, conditions~(i) and (ii) are concerned with edges crossing democratic units, while (iii) with edges connecting peripheral vertices (of different units).
The following proposition characterizes all edges crossing different units when Exchange Rule~\ref{rule:two-units} is not applicable.
\begin{proposition}\label{pro:independent-unit-1}
  Let $\cal U$ be a unit partition on which neither of Exchange Rule \ref{rule:single-unit} and \ref{rule:two-units} is applicable, and let $u_1 u_2$ be an edge between two different units $U_1$ and $U_2$.  If neither $u_1$ nor $u_2$ is a core vertex, then for both $i = 1, 2$, the vertex $u_i$ is either a leaf or the nose of a bull-unit.
\end{proposition}
\begin{proof}
  Suppose for contradiction that $u_1$ is neither a leaf nor a nose.
  We consider the type of $U_1$.  Note that  $u_2$ is not a core of $U_2$.
  If $U_1$ is a democratic unit, then $u_1$ is a non-nose vertex, and condition (i) of  Exchange Rule~\ref{rule:two-units}  holds true.  Otherwise, $U_1$ is a despotic unit, of which $u_1$ is a vertex of a twig, and then one of conditions (i)--(iii) of  Exchange Rule~\ref{rule:two-units} holds true.
\end{proof}

We introduce two more exchange rules to deal with crossing edges between noses of bull-units and leaves of despotic/small units, or edges between two leaves of different units.  One tries to transform a bull-unit to a net-unit, while the other tries to consolidate leaves to make a twig.

\begin{exrule}\label{rule:consolidate-1}
  If there is an edge between the nose $u_1$ of a bull-unit $U_1$ and a leaf $u_2$ of a despotic/small unit $U_2$, move $u_2$ from $U_2$ to $U_1$.
\end{exrule}

\begin{exrule}\label{rule:consolidate-2}
  Let $u_1u_2$ be an edge between two leaves from two units $U_1$ and $U_2$ with $|U_1| \ge |U_2|$.
  \begin{itemize}
  \item If $U_2$ is not a $(1, 4)$-unit, move $u_1$ from $U_1$ to $U_2$.
  \item If $U_2$ is a $(1, 4)$-unit, move $u_2$ from $U_2$ to $U_1$.
  \end{itemize}  
\end{exrule}

Note that Exchange Rule~\ref{rule:consolidate-2} is only well defined when there exists no edges between two leaves of two $(1, 4)$-units; otherwise it would end with a unit with an undefined type.  We say that a unit partition is \emph{general} if (1) each unit in it is one of the types represented by the graphs in Figures~\ref{fig:unit1}--\ref{fig:small-units}; and (2) there is no edge connecting leaves from two different $(1,4)$-units.
We need to make sure that the unit partition remains general after each application of the exchange rules.

\begin{proposition}\label{pro:14unit}
  If $\cal U$ is a general unit partition, then after applying Exchange Rule~\ref{rule:consolidate-1} or \ref{rule:consolidate-2}, we end with a unit with two vertex-disjoint $P_2$'s, or a general unit partition.
\end{proposition}
\begin{proof}
  The proposition is obvious for Exchange Rule~\ref{rule:consolidate-1}: removing a leaf from any despotic/small unit leaves a defined unit.
  Now we consider Exchange Rule~\ref{rule:consolidate-2}.  Note that $U_2$ is not a $(1,4)$-unit, and removing a leaf from it leaves a defined unit.
  On the one hand, $U_1$ becomes a $(d_1(U_1), d_2(U_1)-1)$-unit when $|U_1| \ge 6$, or a small unit otherwise.  On the other hand, if $U_2$ is a $(0,1)$-unit or $(0,3)$-unit, then it becomes a pan-, $C_5$-, bull-, or $(2,0)$-unit; otherwise, it becomes a $(d_1(U_2)+1, d_2(U_2)-1)$-unit.  It is always a defined unit and not a $(1,4)$-unit.
  Since no new $(1,4)$-units are introduced in this step, and since the edges between other units are not impacted, condition (2) of the definition of general unit partition remains satisfied.
  Therefore, the new unit partition is general.
\end{proof}

We conclude the preparation phase by introducing reduced unit partitions. A {\em reduced unit partition} is a general unit partition, on which neither Reduction Rule~\ref{rul:component} nor Exchange Rules \ref{rule:single-unit}--\ref{rule:consolidate-2} are applicable.  The following lemma summarizes the properties of reduced unit partitions.

\begin{lemma}\label{lem:peripheral}
Let $u v$ be an edge crossing two units in a reduced unit partition.  Either both $u$ and $v$ are noses of bull-units, or one of them is a core vertex of some despotic or small unit.
\end{lemma}
\begin{proof}
  Note that a vertex is either a core vertex, or a peripheral vertex, or a vertex in a democratic unit.  If both $u$ and $v$ are peripheral vertices, then condition~(iii) of Exchange Rule~\ref{rule:two-units} and Exchange Rule~\ref{rule:consolidate-2} are applicable.  On the other hand, an edge between a peripheral vertex and a democratic unit would trigger Exchange Rule~\ref{rule:two-units}(i)--(ii) or Exchange Rule~\ref{rule:consolidate-1}.  We are thus left with edges between two democratic units or with one end as a core vertex.  Since Exchange Rule~\ref{rule:two-units}(i) is not applicable, when both $u$ and $v$ are from democratic units, they have to be noses of bull-units.
\end{proof}

\section{Main rules}
This section presents our main exchange rules on the reduced unit partition.  The first two try to move twigs and leaves from a larger unit to a smaller unit through a sequence of intermediary units.
Their purpose is to eliminate all the units with more than six vertices.
In particular, one deals with $(4,0)$-, $(3,1)$-, $(3,0)$- and $(2,2)$-units; while the other deals with $(1,4)$-, $(1,3)$- and $(0,4)$-units.
Note that after applying these two rules, one or more of Exchange Rules~\ref{rule:single-unit}--\ref{rule:consolidate-2} may become applicable again. The last exchange rule and the last reduction rule are concerned with units with six vertices; if there are too many of them, we can find either a larger $P_2$-packing than $\cal P$ or a reducible set.

We start with those units $U$ with $d_1(U) \ge 2$, and the idea is to ``cut'' a twig from such a unit and ``graft'' it onto a smaller unit.  See Figure~\ref{fig:reduce-twig} for illustration.
We say that $t_1 v_2 t_2 \cdots v_\ell$ is a \emph{twig-alternating path} if (1) $v_i$ is a core vertex of $U_i$ for $2 \le i \le \ell$; (2) $t_j$ is a twig of $U_j$ adjacent to $v_{j +1}$ for $1 \le j \le \ell-1$; and (3) $U_i\ne U_j$ when $i \ne j$.
Note that if $d_1(U_1) > 2$, then $U_1$ is a $(4,0)$-, $(3,1)$-, or $(3,0)$-unit.
\begin{exrule}\label{rule:twig}
  Let $\cal U$ be a reduced unit partition.  If there is a twig-alternating path $t_1 v_2 t_2 \cdots v_\ell$ such that
  \begin{enumerate}[(i)]
  \item  $d_1(U_1) > 2$ or $U_1$ is a $(2, 2)$-unit; and
  \item $d_1(U_\ell) = 0$,
  \end{enumerate}
  then for $i = 1, \ldots, \ell - 1$, move twig $t_i$ from unit $U_i$ to unit $U_{i +1}$.
\end{exrule}

\begin{figure}[h]
  \centering
  \begin{tikzpicture}[scale=.5]
    \begin{scope}[shift={(0,0)}]
      \begin{scope}[every node/.style={vertex}]
        \node (a1) at (-1.5, 3) {};
	\node (b1) at (-1.5, 2) {};
	\node (a2) at (-.5, 3) {};
	\node (b2) at (-.5, 2) {};
	\node (a3) at (.5, 3) {};
	\node (b3) at (.5, 2) {};
	\node (b4) at (1.5, 2) {};
      \end{scope}
      \begin{scope}[every node/.style={basev}]
	\node (u) at (0, 0) {};
      \end{scope}
      \draw (a1) -- (b1) -- (u);
      \draw (a2) -- (b2) -- (u);
      \draw[blue, very thick] (a3) -- (b3) -- (u);
      \draw (b4) -- (u);
      \node at (0, -1) {$U_1$};
    \end{scope}

    \begin{scope}[shift={(4,0)}]
      \begin{scope}[every node/.style={vertex}]
	\node (c1) at (-.75, 3) {};
	\node (d1) at (-.75, 2) {};
	\node (c2) at (.75, 3) {};
	\node (d2) at (.75, 2) {};
      \end{scope}
      \begin{scope}[every node/.style={basev}]
	\node (v) at (0, 0) {};
      \end{scope}
      \draw (c1) -- (d1) -- (v);
      \draw (c2) -- (d2) -- (v);
      \node at (0, -1) {$U_2$};
    \end{scope}

    \begin{scope}[shift={(7,0)}]
      \begin{scope}[every node/.style={vertex}]
	\node (e1) at (-.75, 3) {};
	\node (e2) at (-.75, 2) {};
	\node (e3) at (0, 2) {};
	\node (e4) at (.75, 2) {};
      \end{scope}
      \begin{scope}[every node/.style={basev}]
	\node (w) at (0, 0) {};
      \end{scope}
      \draw (e1) -- (e2) -- (w);
      \draw (e3) -- (w) -- (e4);
      \node at (0, -1) {$U_3$};
    \end{scope}

    \begin{scope}[shift={(10, 0)}]
      \begin{scope}[every node/.style={vertex}]
	\node (f1) at (-.75, 2) {};
	\node (f2) at (0, 2) {};
	\node (f3) at (.75, 2) {};
      \end{scope}
      \begin{scope}[every node/.style={basev}]
	\node (x) at (0, 0) {};
      \end{scope}
      \draw (f1) -- (x);
      \draw (f2) -- (x) -- (f3);
      \node at (0, -1) {$U_4$};
    \end{scope}
    \draw[dashed, blue, thick] (b3) edge (v) (c1) edge (w) (e2) edge (x);
    \node at (13,  1.5) {$\Rightarrow$};
  \end{tikzpicture}
  \qquad
  \begin{tikzpicture}[scale=.5]
    \begin{scope}[shift={(0,0)}]
      \begin{scope}[every node/.style={vertex}]
	\node (a1) at (-1.5, 3) {};
	\node (b1) at (-1.5, 2) {};
	\node (a2) at (-.5, 3) {};
	\node (b2) at (-.5, 2) {};
	\node (b4) at (1.5, 2) {};
      \end{scope}
      \begin{scope}[every node/.style={basev}]
	\node (u) at (0, 0) {};
      \end{scope}
      \draw (a1) -- (b1) -- (u);
      \draw (a2) -- (b2) -- (u);
      \draw (b4) -- (u);
      \node at (0, -1) {$U'_1$};
    \end{scope}

    \begin{scope}[shift={(4,0)}]
      \begin{scope}[every node/.style={vertex}]
	\node (c1) at (-.75, 3) {};
	\node (d1) at (-.75, 2) {};
	\node (c2) at (.75, 3) {};
	\node (d2) at (.75, 2) {};
      \end{scope}
      \begin{scope}[every node/.style={basev}]
	\node (v) at (0, 0) {};
      \end{scope}
      \draw (c1) -- (d1) -- (v);
      \draw (c2) -- (d2) -- (v);
      \node at (0, -1) {$U'_2$};
    \end{scope}

    \begin{scope}[shift={(7,0)}]
      \begin{scope}[every node/.style={vertex}]
	\node (e1) at (-.75, 3) {};
	\node (e2) at (-.75, 2) {};
	\node (e3) at (0, 2) {};
	\node (e4) at (.75, 2) {};
      \end{scope}
      \begin{scope}[every node/.style={basev}]
	\node (w) at (0, 0) {};
      \end{scope}
      \draw (e1) -- (e2) -- (w);
      \draw (e3) -- (w) -- (e4);
      \node at (0, -1) {$U'_3$};
    \end{scope}

    \begin{scope}[shift={(10, 0)}]
      \begin{scope}[every node/.style={vertex}]
	\node (f4) at (-1.5, 3) {};
	\node (f5) at (-1.5, 2) {};
	\node (f1) at (-.5, 2) {};
	\node (f2) at (.5, 2) {};
	\node (f3) at (1.5, 2) {};
      \end{scope}
      \begin{scope}[every node/.style={basev}]
	\node (x) at (0, 0) {};
      \end{scope}
      \draw[blue, very thick] (f4) -- (f5) -- (x);
      \draw (f1) -- (x);
      \draw (f2) -- (x) -- (f3);
      \node at (0, -1) {$U'_4$};
    \end{scope}
    \draw[dashed, blue, thick] (d1) edge (u)  (v) edge (e1) (w) edge (f5);
  \end{tikzpicture}
  \caption{Illustration for Exchange Rule~\ref{rule:twig}.  The left is the original graph and the right the graph after applying the exchange rule.  The first unit loses a twig (denoted by two thick edges in the left), the last unit gains a twig (denoted by two thick edges in the right), while the type of all the intermediary units remain unchanged.  In a sense, after a chain of operations, ``a twig is moved from the first unit to the last one.''}
  \label{fig:reduce-twig}
\end{figure}

Let $U'_\ell$ be the new unit obtained by adding $t_{\ell - 1}$ to $U_\ell$. Note that if $U_\ell$ is a $(0, 4)$-unit, then $U'_\ell$ will be a $(1, 4)$-unit.  Indeed, Exchange Rule~\ref{rule:twig} is the only rule that can introduce $(1,4)$-units.
After applying Exchange Rule~\ref{rule:twig}, we may need to apply Exchange Rule~\ref{rule:consolidate-1} or \ref{rule:consolidate-2}, so we need to make sure that the resulting unit partition is still general.
\begin{proposition}\label{pro:twig}
  After applying Exchange Rule~\ref{rule:twig} for a twig-alternating path $P = t_1 v_2 t_2 \cdots x_\ell$, the following hold.
  \begin{enumerate} [(i)]
  \item $U_1$ becomes a $(d_1(U_1)-1, d_2(U_1))$-unit.
  \item For $2 \le i \le \ell -1$, both $d_1(U_i)$ and $d_2(U_i)$ remain unchanged.
  \item If $|U'_\ell| > 6$ and $d_2(U_\ell) \ge 2$, then $d_1(U'_\ell) = 1$ and $d_2(U'_\ell) = d_2(U_{\ell})$. 
  \item Either $U'_\ell$ contains two vertex-disjoint $P_2$'s, or the resulting unit partition is still general.
  \end{enumerate}
\end{proposition}
\begin{proof}
  The first three assertions are immediate from the rule.  Note that if $U_\ell$ was a $(0, 1)$-unit, then it may happen that the new unit contains two vertex-disjoint $P_2$'s.  We now argue that if this does not happen, then the resulting unit partition is still general.  Note that $U_\ell$ is $(0,4)$-, $(0,3)$-, $(0,1)$- or $(0,0)$-unit.  All the $\ell$ units are still well defined, in particular, $U'_\ell$ is a democratic unit or a despotic unit of type $(1, 2)$, $(1, 3)$, $(1, 4)$, $(2, 0)$, or $(2, 1)$.  In the case that it is a $(1, 4)$-unit, $U_\ell$ was a $(0, 4)$-unit.  Because the unit partition was reduced before we apply Exchange Rule~\ref{rule:twig}, there cannot be an edge between leaves of $U_\ell$ and leaves from another $(1, 4)$-unit by Lemma~\ref{lem:peripheral}.  Hence the unit partition remains general.
\end{proof}

If some unit has at least six vertices and at least two twigs, but we cannot move any twig away from it with Exchange Rule~\ref{rule:twig}, then we can find a reducible set.  For this purpose, we build an auxiliary bipartite graph: 
\begin{itemize}
\item for each core vertex, introduce a node into $L$;
	
\item for each twig, introduce a node into $R$; and
	
\item add an edge between a node $x \in L$ and a node $y \in R$ if the core vertex represented by $x$ is adjacent to the twig represented by $y$.
\end{itemize}

\begin{lemma}\label{lem:twig-crown}
  Let $\cal U$ be a reduced unit partition.  If there exists a unit $U$ in $\cal U$ with $d_1(U) = d_2(U) = 2$ or $d_1(U) > 2$, but Exchange Rule~\ref{rule:twig} is not applicable, then we can find a reducible set in polynomial time.
\end{lemma}
\begin{proof}
  We take all the twig-alternating paths with $t_1$, the first twig, from $U$, and let $C$ denote the set of vertices in the twigs in them.  By Proposition~\ref{pro:extra-edges} and Lemma~\ref{lem:peripheral}, $G[C]$ consists of edge components and all vertices in $N(C)$ are core vertices.  Since Exchange Rule~\ref{rule:twig} is not applicable, for each core vertex in $N(C)$, the unit containing it has at least one twig.  Hence $C$ is a reducible set.
\end{proof}
	
We next migrate leaves in a similar way as Exchange Rule~\ref{rule:twig}.  We say that $t_1 v_2 t_2 \cdots v_\ell$ is a \emph{leaf-alternating path} if (1) $v_i$ is the core vertex of $U_i$ for $2 \le i \le \ell$; (2) $t_j$ is a leaf of $U_j$ adjacent to $v_{j + 1}$ for $1 \le j \le \ell-1$; and (3) $U_i\ne U_j$ when $i \ne j$.
\begin{exrule}\label{rule:leaf}
  In a reduced unit partition without $(4,0)$-, $(3,1)$-, $(3,0)$-, or $(2,2)$-units, if there is a leaf-alternating path $t_1 v_2 t_2 \cdots v_\ell$ where
  \begin{enumerate} [(i)]
  \item $U_1$ is not a $(0, 3)$-unit and $d_2(U_1) \ge 3$; and
  \item $d_2(U_\ell) \le 1$.
  \end{enumerate}
  then for $i = 1, \ldots, \ell - 1$, move leaf $t_i$ from unit $U_i$ to unit $U_{i +1}$.
\end{exrule}

Let $U'_\ell$ be the new unit obtained by adding $t_{\ell - 1}$ to $U_\ell$.
Again, we need to make sure that the resulting unit partition is still general.

\begin{proposition}\label{pro:leaf}
  After applying Exchange Rule~\ref{rule:leaf} for a leaf-alternating path $P = t_1 v_2 t_2 \cdots v_\ell$, the following hold.
  \begin{enumerate} [(i)]
  \item $U_1$ becomes a $(d_1(U_1), d_2(U_1)-1)$-unit.
  \item For $2 \le i \le \ell -1$, both $d_1(U_i)$ and $d_2(U_i)$ remain unchanged.
  \item If $|U_\ell| > 4$, then $d_1(U'_\ell) = d_1(U_\ell)$ and $d_2(U'_\ell) = d_2(U_{\ell}) + 1$.
  \item The resulting unit partition is still general.
  \end{enumerate}
\end{proposition}
\begin{proof}
The first two assertions are immediate from the rule.
Note that Exchange Rule~\ref{rule:leaf} is only applied on a unit partition with $d_1(U)\le 3$ for all units $U$.  If $|U_\ell| > 4$, then $U_\ell$ has to be a $(2, 1)$- or $(2, 0)$-unit, and $U'_\ell$ is a $(2, 2)$- or $(2, 1)$-unit respectively.  We are then left with the case $U_\ell$ being a small unit, i.e., $|U_\ell|\le 4$.  This proves assertion (iii).
Note that Exchange Rule~\ref{rule:leaf} may introduce new edges between leaves, by making core vertices of $(0,1)$- or $(0,0)$-units peripheral.
If $U_\ell$ is a $(0,0)$-unit, then $U'_\ell$ is a small unit on four vertices.
If $U_\ell$ is a $(0,1)$-unit, then $U'_\ell$ may be a $(2, 0)$-, $(1, 2)$-, pan-, or bull-unit.
Noting that no $(1, 4)$-unit is produced or impacted during this operation, we can conclude that the resulting unit partition remains general.  This concludes the proof.
\end{proof}

A lemma similar to Lemma~\ref{lem:twig-crown} holds for Exchange Rule~\ref{rule:leaf}.  We build an auxiliary bipartite graph: 
\begin{itemize}
\item for each core vertex, introduce a node into $L$;
\item for each leaf, introduce a node into $R$; and
\item add an edge between a node $x \in L$ and a node $y \in R$ if the core vertex represented by $x$ is adjacent to the leaf represented by $y$.
\end{itemize}

\begin{lemma}\label{lem:leaf-crown}
  Let $\cal U$ be a reduced unit partition without $(4,0)$-, $(3,1)$-, or $(3,0)$-units.  If there exists in $\cal U$ a unit $U$ with (i) $d_2(U) > 3$, or (ii) $d_1(U) > 0$ and $d_2(U) = 3$, but Exchange Rule~\ref{rule:leaf} is not applicable, then we can find a reducible set in polynomial time.
\end{lemma}
\begin{proof}

	We take all the leaf-alternating paths with $t_1$, the first leaf, from $U$, and let $C$ denote the set of leaves in them.  By Proposition~\ref{pro:extra-edges} and Lemmas~\ref{lem:peripheral}, $G[C]$ consists of isolated vertices and all vertices in $N(C)$ are core vertices.  Since Exchange Rule~\ref{rule:leaf} is not applicable, for each core vertex in $N(C)$, the unit containing it has at least two leaves.  Hence $C$ is a reducible set.
\end{proof}

After application of Exchange Rule~\ref{rule:twig} or \ref{rule:leaf}, a core vertex of a $(0, 1)$-unit or $(0, 0)$-unit may become a peripheral vertex, thereby turning a reduced unit partition back to a general unit partition.  Therefore we may need to reapply Exchange Rules~\ref{rule:single-unit}--\ref{rule:consolidate-2}.   Moreover, Exchange Rule~\ref{rule:leaf} may turn a $(2,1)$-unit into a $(2,2)$-unit, and hence trigger the application of Exchange Rule~\ref{rule:twig}.  We are able to show that these two exchange rules can be applied a polynomial number of times; we defer the issue of the running time to the next section.

If neither Reduction Rule~\ref{rul:component} nor Exchange Rules~\ref{rule:single-unit}--\ref{rule:leaf} is applicable, then the only despotic units that can exist in a general unit partition are $(2,1)$-, $(2,0)$- and $(1,2)$-units.  Therefore, of all units, only net-units and $(2, 1)$-units have more than five vertices.  We now try to bound the number of these two type of units by the number of other smaller units.  We use $s_{\mathrm{net}}$, $s_{\mathrm{pan}}$, $s_{C_5}$ and $s_{\mathrm{bull}}$ to denote the number of, respectively, net-, pan-, $C_5$- and bull-units, and use $s_{a, b}$ to denote the number of $(a, b)$-units.

We construct an auxiliary bipartite graph $B_4$:
\begin{itemize}
\item for each net-unit and each twig of a despotic unit, add a node into $L$;
\item for each vertex not in any net-unit or twig, add a node into $R$; and
\item for two nodes $x\in L$ and $y \in R$, add an edge $xy$ if the vertex represented by $y$ is adjacent to the net-unit or twig represented by $x$ in $G$.
\end{itemize}
Informally speaking, we are here viewing a net-unit as ``a unit with a removable twig.''  Note that there is no isolated node in $L$ (as otherwise Reduction Rule~\ref{rul:component} should be applicable), and each node in $N(L)$ represents a core vertex.

\begin{exrule}~\label{rule:final}
  Let $\cal U$ be a reduced unit partition in which all despotic units are $(2,1)$-, $(2,0)$-, and $(1,2)$-units.  If $s_{\mathrm{net}} + s_{2, 1} + s_{2, 0} > s_{0, 3} + s_{0, 1} + s_{0, 0}$ and there is a matching $M$ of $B_4$ that saturates $L$, then find a larger $P_2$-packing than $\cal U$ as follows.

  \begin{enumerate}[(i)]
  \item For each net-unit $U$, introduce two $P_2$'s from $G[U \cup \{v\}]$, where $v$ is the vertex represented by the node matched to $U$ in $M$.
  \item For each twig $t$, introduce a $P_2$ together with the vertex $v$ represented by the node matched to $t$ in $M$.
  \item From each other unit not involved in $M$, take an arbitrary $P_2$.
  \end{enumerate}
\end{exrule}
\begin{lemma}
  Exchange Rule~\ref{rule:final} is correct: (1) The $P_2$'s can be found; (2) they are vertex-disjoint; and (3) the number is larger than $|\cal U|$.
  \end{lemma}
\begin{proof}
  Assertion (1) is straightforward.  The vertices $v$ matched to different net-units and twigs are different. Hence, no vertex is shared by two $P_2$'s in (i) and (ii).  On the other hand, the $P_2$'s found in (iii) use vertices not involved in those of (i) and (ii).  This verifies assertion (2).  Since the matching $M$ saturates $L$, we have introduced two $P_2$'s for each net-, $(2,1)$-, and $(2,0)$-unit, and one $P_2$ for each $(1,2)$-unit.  On the other hand, by Lemma~\ref{lem:peripheral}, the vertex $v$ in both (i) and (ii) has to be a core vertex.  In other words, no pan-, $C_5$-, and bull-units are involved in $M$.  Therefore, the total number of $P_2$'s we can find is at least

  \begin{align*}
    &2(s_{\mathrm{net}} + s_{2, 1} + s_{2, 0}) + s_{1,2} + (s_{\mathrm{pan}} + s_{C_5} + s_{\mathrm{bull}})
    \\
    >&  (s_{\mathrm{net}} + s_{2, 1} + s_{2, 0}) + (s_{0, 3} + s_{0, 1} + s_{0, 0}) + s_{1,2} + s_{\mathrm{pan}} + s_{C_5} + s_{\mathrm{bull}}
    \\
    =&
       (s_{\mathrm{net}} + s_{\mathrm{pan}} + s_{C_5} + s_{\mathrm{bull}}) + (s_{2, 1} + s_{2, 0} + s_{1,2}) + (s_{0, 3} + s_{0, 1} + s_{0, 0})
    \\
    =& |\cal U|.
  \end{align*}
  Note that the last equality is true because the only despotic units in $\cal U$ are $(2,1)$-, $(2,0)$-, and $(1,2)$-units.
\end{proof}

If no matching in $B_4$ saturates $L$, then we can use Lemma~\ref{lem:hopcroft-matching} to find in polynomial time a subset $L'\subseteq L$ and a matching of $B_4$ between $N(L')$ and $L'$ that saturates $N(L')$.
Now we apply our second non-trivial reduction rule, where the ``reducible set" contains net-units and edge components.  Each net-unit contributes a $P_2$ and an extra edge.

\begin{redrule}\label{rul:net-crown}
  For a node set $L'\subseteq L$ of $B_4$, let $s'_{\mathrm{net}}$ be the number of net-units represented by $L'$, and $X$ be the set of vertices in the net-units or twigs represented by $L'$. If there exists a matching of $B_4$ between $N(L')$ and $L'$ that saturates $N(L')$, then remove $X \cup N(X)$ and decrease $k$ by $s'_{\mathrm{net}} + |N(X)|$.
\end{redrule}
\begin{lemma}\label{lem:net-crown}
  Reduction rule~\ref{rul:net-crown} is safe: $\mathtt{opt}(G) = \mathtt{opt}(G- (X \cup N(X))) + s'_{\mathrm{net}} + |N(X)|$.
\end{lemma}
\begin{proof}
  Let $G' = G- (X \cup N(X))$. Note that $N(X)$ is precisely the set of vertices represented by $N(L')$. Let $v$ be a vertex in $N(X)$.  If its representing node is matched to a twig, then $v$ and the twig together make a $P_2$.  Otherwise, it is matched to a net-unit $U$; i.e., $U$ is adjacent to $v$.  We can find two adjacent vertices in $U$ such that they together with $v$ make another $P_2$, and their removal from $U$ leaves a $P_2$.  Thus the number of $P_2$'s we can find in $G[X\cup N(X)]$ is $|N(X)| + s'_{\mathrm{net}}$.  Therefore, $\mathtt{opt}(G) \ge \mathtt{opt}(G- (X \cup N(X))) + s'_{\mathrm{net}} + |N(X)|$.

  Any $P_2$-packing of $G$ contains at most $|N(X)|$ vertex-disjoint $P_2$'s involving vertices in $N(X)$, hence $\mathtt{opt}(G) \le \mathtt{opt}(G-N(X)) + |N(X)|$.  Each component in the subgraph of $G - N(X)$ induced on $X$ is either a net-unit or a twig represented by a node in $L'$. Thus $\mathtt{opt}(G-N(X)) = \mathtt{opt}(G- (X \cup N(X))) + s'_{\mathrm{net}}$, and $\mathtt{opt}(G) \le \mathtt{opt}(G-(X \cup N(X))) + s'_{\mathrm{net}} + |N(X)|$.
\end{proof}

\begin{corollary}\label{cor:big-units}
  Let $\cal U$ be a reduced unit partition where $(2,1)$-, $(2,0)$- and $(1,2)$-units are the only despotic units.  If $s_{\mathrm{net}} + s_{2, 1} + s_{2, 0} > s_{0, 3} + s_{0, 1} + s_{0, 0}$, then at least one of Exchange Rule~\ref{rule:final} and Reduction Rule~\ref{rul:net-crown} is applicable.
\end{corollary}

\section{The algorithm}

We now summarize the kernelization algorithm in Figure~\ref{fig:alg} and use it to prove Theorem~\ref{thm:main}, for which we need to show its correctness and analyze its running time.

\begin{figure}[h!]
  \centering
  \begin{tikzpicture}
    \path (0,0) node[text width=.8\textwidth, inner xsep=20pt, inner ysep=10pt] (a) {
    \begin{minipage}[t!]{\textwidth} \small
  {\sc Input}: a graph $G$ and an integer $k$.
  \\
  {\sc Output}: a graph $G'$ with $|V(G')| \le 5k$ and an integer $k'$.

  \begin{tabbing}
    AAa\=Aaa\=aaa\=Aaa\=MMMMMMAAAAAAAAAAAAA\=A \kill

    1.\> {\bf if} $k \le 0$ {\bf then return} a trivial yes-instance.
    \\
    2.\> {\bf if} $|V(G)| \le 5 k$ {\bf then return} ($G, k$).
    \\
    3.\> Apply Reduction Rule~\ref{rul:component} exhaustively.
    \\
    4.\> Set $\cal P$ to be an arbitrary maximal $P_2$-packing.
    \\
    5.\> {\bf if} $|{\cal P}| \ge k$ {\bf then return} a trivial yes-instance.
    \\
    6.\> {\bf if} the condition of Lemma~\ref{lem:matching} is true {\bf then}
    \\
    \>\> apply Reduction Rule~\ref{rul:p2-crown}; {\bf goto 1}.
    \\
    7.\> {\bf else} build a unit partition {$\cal U$};
    \\
    8.\> {\bf if} Exchange Rule~\ref{rule:single-unit} or \ref{rule:two-units} is applicable {\bf then}
    \\
    \>\> apply it; {\bf goto 5}.
    \\
    9.\> {\bf if} Exchange Rule~\ref{rule:consolidate-1} or \ref{rule:consolidate-2} is applicable {\bf then}
    \\
    \>\> apply it; {\bf goto 8}.
    \\
    \comment{  Now the unit partition is reduced.}
    \\
    10.\> {\bf if} there is a $(4,0)$-, $(3,1)$-, $(3,0)$- or $(2,2)$-unit {\bf then}
    \\
    \>\> {\bf if} Exchange Rule~\ref{rule:twig} is applicable {\bf then}
    \\
    \>\>\> apply it; {\bf goto 8}.
    \\
    \>\> {\bf else} apply Reduction Rule~\ref{rul:p2-crown} (Lemma~\ref{lem:twig-crown}); {\bf goto 1};
    \\
    11.\> {\bf if} there is a $(1,4)$-, $(1,3)$- or $(0,4)$-unit {\bf then}
    \\
    \>\> {\bf if} Exchange Rule~\ref{rule:leaf} is applicable, {\bf then}
    \\
    \>\>\> apply it; {\bf goto 8}.
    \\
    \>\> {\bf else} apply Reduction Rule~\ref{rul:p2-crown} (Lemma~\ref{lem:leaf-crown}); {\bf goto 1}.
    \\
    12.\> {\bf if} the number of net-, $(2,1)$-, $(2,0)$-units $>$ the number of small units {\bf then}
    \\
    \>\> {\bf if} Exchange Rule~\ref{rule:final} is applicable {\bf then}
    \\
    \>\>\> apply it; {\bf goto 5}.
    \\
    \>\> {\bf else} apply Reduction Rule~\ref{rul:net-crown} (Corollary~\ref{cor:big-units}); {\bf goto 1}.
    \\
    13.\> {\bf return} $(G, k)$.
\end{tabbing}
    \end{minipage}
  };
    \draw[draw=gray!60] (a.north west) -- (a.north east) (a.south west) -- (a.south east);
  \end{tikzpicture}
  \caption{A summary of our algorithm.  A trivial yes-instance can be an empty graph and $k = 0$.
    Note that if the condition of step 8 is satisfied, we will execute it, and then go back to step 5.  It is similar for steps~9--12.  In particular, we will never execute two of steps~9--12 consecutively.}
\label{fig:alg}
\end{figure}

In general, we apply the exchange and reduction rules in the order they are presented; more specifically, we apply a high numbered rule only when no previous rule is applicable.  But the application of a later rule may enable a previous one.  As a consequence, the analysis of running time is focused on the maximum number of times they are applied.  The first two reduction rules and the first two exchange rules will not concern us, because they either remove vertices from the graph or increase the size of the $P_2$-partition $\cal P$: Note that the size of $\cal P$ never decreases, and the algorithm is terminated when it reaches $k$.  For the same reason, Exchange Rule~\ref{rule:final} and Reduction Rule~\ref{rul:net-crown} can be applied at most $k$ and $n$ times respectively.

The following invariant is immediate from the definitions of our exchange rules: A net-unit can only be changed by the reduction rules or Exchange Rules~\ref{rule:single-unit}, \ref{rule:two-units}, and~\ref{rule:final}.
Since each application of Exchange Rule~\ref{rule:consolidate-1} introduces a new net-unit, it cannot be applied more than $k$ times without triggering previous rules.

We are thus left with Exchange Rules~\ref{rule:consolidate-2}--\ref{rule:leaf}, which are our main trouble.  If the transformation among the units is acyclic (under the condition that both $G$ and $\cal P$ are unchanged), then we can find an invariant similar as above.  This is unfortunately not this case.  For example, Exchange Rule~\ref{rule:consolidate-2} may turn a $(1, 2)$-unit into a $(0, 1)$-unit, but Exchange Rule~\ref{rule:leaf} may change it back.  The way we surmount the obstacle is to define a measure that can only \emph{decrease} when $G$ or $\cal P$ is changed.  Since its value increases when any of Exchange Rules~\ref{rule:consolidate-2}--\ref{rule:leaf} is applied, but it never exceeds  $O(k)$, we bound the number of times they can be applied.

Let $s_{\mathrm{democratic}}$, $s_{\ge 1, *}$, $s_{*, \ge 1}$, and $s_{*, \ge 2}$ denote the numbers of, respectively, democratic units, units with $d_1(U) \ge 1$, units with $d_2(U) \ge 1$, and units with $d_2(U) \ge 2$, and let $s_{\mathrm{twig}}$ denote the total number of twigs among all despotic units.
We define a measure
\[
  s = 51 s_{\mathrm{democratic}} + 6s_{\mathrm{twig}} + (s_{\ge 1, *} + s_{*, \ge 2} + s_{*, \ge 1}) + 4 s_{0, 1}.
\]
Note that a unit may contribute more than one term to this measure.  For example, a $(1, 2)$-unit contributes to the four terms $s_{\mathrm{twig}}$, $s_{\ge 1, *}$, $s_{*, \ge 2}$, and $s_{*, \ge 1}$.
See Table~\ref{tbl:table0} for the contributions to $s$ of different units.

\begin{table}[h]
  \caption{Total contributions of despotic and small units to the measure $s$.}
  \label{tbl:table0}
  \begin{tabular}{c|c|c|c|c|c|c|c|c|c|c|c|c}
   $(4, 0)$ & $(3, 1)$ & $(3, 0)$ & $(2, 2)$ & $(2, 1)$ & $(2, 0)$ & $(1, 4)$ & $(1, 3)$ & $(1, 2)$ & $(0, 4)$ & $(0, 3)$ & $(0, 1)$ & $(0, 0)$ \\ \hline
   25       & 20       & 19       & 15       & 14       & 13       & 9        & 9        & 9        & 2        & 2        & 5        & 0
  \end{tabular}
\end{table}

\begin{lemma}\label{lem:increasing}
  After each application of Exchange Rules~\ref{rule:consolidate-2}--\ref{rule:leaf}, either a unit with two vertex-disjoint $P_2$'s is produced, or $s$ increases.
\end{lemma}
\begin{proof}
  Since an application of Exchange Rules~\ref{rule:consolidate-2}--\ref{rule:leaf} changes the types of two units, $s$ can decrease by at most $25 *2 = 50$. If the number of democratic units increases, then $s$ increases.
  Hence, we may assume that no democratic unit is introduced during this application.
	
  Let $U_-$ and $U_+$ denote respectively the two units losing and gaining vertices in this operation.
  If they do not become democratic, and $\cal P$ does not increase, then they become despotic or small units after the operation.
  Note that when $U_-$ is a $(0, 1)$-unit before applying Exchange Rule~\ref{rule:consolidate-2}, $U_+$ has to be a $(0, 1)$- or $(0, 3)$-unit as well.  Therefore, it is immediate from Table~\ref{tbl:table1} that  the increment is always strictly larger than the decrement.  In other words, $s$ always increases with this operation.
\end{proof}

\begin{table}[h]
  \caption{The change of measure incurred by applying Exchange Rules~\ref{rule:consolidate-2}--\ref{rule:leaf}; e.g., the measure decreases by one if a vertex is moved \emph{out of} a $(3, 1)$-unit by Rule~\ref{rule:consolidate-2}, while it increases by five if a vertex is moved \emph{into} a $(3, 1)$-unit by Rule~\ref{rule:consolidate-2}.  A $\bot$ means that this type of units cannot participate in this operation in the specific role; e.g., we are not allowed to move any vertex {out of} a $(3, 0)$-unit using Rule~\ref{rule:consolidate-2}.  Note that in none of the rules, $U_-$ can be a $(0, 0)$-unit, or $U_+$ a $(4, 0)$-, $(3, 0)$- or $(1, 4)$-unit.}
  \label{tbl:table1}
  \begin{tabular}{c|c|c|c|c|c|c|c|c|c|c|c|c}
    $U_-$  & $(4, 0)$ & $(3, 1)$ & $(3, 0)$ & $(2, 2)$ & $(2, 1)$ & $(2, 0)$ & $(1, 4)$   & $(1, 3)$  & $(1, 2)$ & $(0, 4)$ & $(0, 3)$   & $(0, 1)$
    \\ \hline
    Rule 4 & $\bot$        & $-1$       & $\bot$        & \multicolumn{2}{c|}{$-1$}       & $\bot$        & \multicolumn{2}{c|}{0} & $-4$       & 0        & $-2$         & $-5$
    \\ \hline
    Rule 5 & \multicolumn{4}{c|}{$-6$}                   & \multicolumn{8}{c}{$\bot$}                                                                      \\ \hline
    Rule 6 & \multicolumn{6}{c|}{$\bot$}                                          & \multicolumn{2}{c|}{0} & $\bot$        & 0        & \multicolumn{2}{c}{$\bot$}
  \end{tabular}

  \begin{tabular}{c|c|c|c|c|c|c|c|c|c|c|c}
    $U_+$  & $(3, 1)$   & $(2, 2)$  & $(2, 1)$   & $(2, 0)$   & $(1, 4)$ & $(1, 3)$ & $(1, 2)$ & $(0, 4)$   & $(0, 3)$   & $(0, 1)$ & $(0, 0)$ \\ \hline
    Rule 4 & \multicolumn{3}{c|}{ +5}  & \multicolumn{2}{c|}{$\bot$}         & +6              & +5       & \multicolumn{2}{c|}{ +7}  & +8     & $\bot$        \\ \hline
    Rule 5 & \multicolumn{7}{c|}{$\bot$}                                                                              & \multicolumn{2}{c|}{+7} & +9      & +9/+13  \\ \hline
    Rule 6 & \multicolumn{2}{c|}{$\bot$} & \multicolumn{2}{c|}{+1} & \multicolumn{5}{c|}{$\bot$}                                                     & +4/+8    & +2/+5
  \end{tabular}
\end{table}

Putting them together, we are now ready to prove our main theorem.

\begin{proof}[Proof of Theorem~\ref{thm:main}]
  We demonstrate the correctness first.  Most of the steps are straightforward, except steps 10--12, for which we need to show that the unit partition is reduced when any of the steps 10--12 is executed.
  Let ${\cal U}$ be the unit partition before the execution of one of steps 10--12.  By the algorithm,  neither Reduction Rule~\ref{rul:component} nor Exchange Rules~\ref{rule:single-unit}--\ref{rule:consolidate-2} are applicable on ${\cal U}$: Note that Reduction Rule~\ref{rul:component} can become applicable again after another reduction rule, and we check immediately for each of them (steps 6, 10, 11, 12).  It hence remains to verify that ${\cal U}$ is a general unit partition.

  After exhaustive executions of step 8,
  the unit partition created by step 7 becomes general,
  and contains no $(1,4)$-units by Proposition~\ref{pro:raw-unit-partition}.
  By Propositions~\ref{pro:14unit}, \ref{pro:twig}, and \ref{pro:leaf},
  the unit partition ${\cal U}$ remains general.

  We then calculate the size of the kernel returned by our algorithm and show $|V(G)| \le 5k$.  We have nothing to worry if it is from step~1, 2, or 5.  In the rest the kernel is returned from step~13.  It is accompanied with a reduced unit partition in which (1) all non-democratic units satisfy $d_1(U) \le 2$ and $d_1(U) + d_2(U) \le 3$; and (2) $s_{\mathrm{net}} + s_{2, 1} + s_{2, 0} \le s_{0, 3} + s_{0, 1} + s_{0, 0}$.  The number of vertices in the resulting graph is
  \begin{align*}
    & 6 s_{\mathrm{net}} + 5(s_{\mathrm{pan}} + s_{C_5} + s_{\mathrm{bull}}) + 6s_{2, 1} + 5(s_{2, 0} + s_{1, 2}) + 4(s_{0, 3} + s_{0, 1}) + 3 s_{0, 0}
    \\
    \le & 5 (s_{\mathrm{net}} + s_{2, 1} + s_{\mathrm{pan}} + s_{C_5} + s_{\mathrm{bull}} + s_{2, 0} + s_{1, 2}) + (s_{\mathrm{net}} + s_{2, 1}) + 4(s_{0, 3} + s_{0, 1} + s_{0, 0})
    \\
    \le & 5 (s_{\mathrm{democratic}} + s_{2, 1} + s_{2, 0} + s_{1, 2}) + (s_{0, 3} + s_{0, 1} + s_{0, 0}) + 4(s_{0, 3} + s_{0, 1} + s_{0, 0})
    \\
    =& 5 (s_{\mathrm{democratic}} + s_{2, 1} + s_{2, 0} + s_{1, 2} + s_{0, 3} + s_{0, 1} + s_{0, 0})
    \\
    <& 5 k.
  \end{align*}

  It remains to show that the algorithm terminates in polynomial time.

  There are three destinations for the goto statements in the algorithms, namely steps 1, 5, and 8.  Note that before going back to step 1, from step 6, 10, 11, or 12, one of Reduction Rules~\ref{rul:p2-crown} and~\ref{rul:net-crown} is applied; in other words, each execution of step 1 is with a decreased $k$.  Since $k$ never increases during the algorithm, the algorithm can execute step 1 at most $k$ times.
  Steps 2--5 are run in sequence after step~1, unless an exit condition is satisfied and the algorithm is terminated (step~2 or 5).  
  Each application of Reduction Rule~\ref{rul:component} removes at least one vertex from the graph.  Hence, it can be checked $n + k$ times and applied $n$ times.

The algorithm may go back to step~5 from step 8 or 12; in either case, we have a strictly larger $P_2$-packing $\cal P$.  Since the algorithm terminates when the cardinality of $\cal P$ reaches $k$, it can reach step~5 at most $k$ times without executing step~1 (by a goto statement).  In total, step~5 can be run at most $k^2$ times.
The total times that step~6 or 8 are reached with the condition satisfied, hence going back to step 1 or 5 respectively, have been bounded.  Otherwise, they run in sequel till step~9.

Since each application of Exchange Rule~\ref{rule:consolidate-1} introduces a new net-unit, it cannot be applied more than $k$ times without triggering previous rules, hence going back to step~1 or 5.

According to Lemma~\ref{lem:increasing}, if the algorithm goes back to step~8 from any of step~9--12, either Exchange Rule~\ref{rule:single-unit} becomes applicable (hence we apply it and go back to step~5), or the measure is increased.  Since the measure is $O(k)$, the four steps (9--12) cannot be run more than $O(k)$ times without reaching step~1 or 5.  Therefore, the total number of times they can be run is $O(k^3)$.

If the condition of step~12 is true, then the algorithm goes back to either step~5 or step~1.  Hence it can be run at most $k + k + k^2 = O(k^2)$ times.

Therefore, each rule is checked and applied $O(n + k^3)$ times.  Since each of them can be checked and applied in polynomial time, the total running time is polynomial on $n$ and $k$.  This concludes the analysis of running time and the proof.
\end{proof}

{
  \bibliographystyle{plainurl}

}
\end{document}